\documentclass[a4paper]{amsproc}

\usepackage[utf8]{inputenc}
\usepackage{algorithm,algorithmic}

\newtheorem{theorem}{Theorem}[section]

\newtheorem{prop}[theorem]{Proposition}

\theoremstyle{definition}
\newtheorem{definition}[theorem]{Definition}
\newtheorem{example}[theorem]{Example}

\theoremstyle{remark}
\newtheorem{remark}[theorem]{Remark}

\numberwithin{equation}{section}

\newcommand{\abs}[1]{\lvert#1\rvert}

\newcommand{\Z}{\mathbb{Z}}
\newcommand{\N}{\mathbb{N}}
\newcommand{\Q}{\mathbb{Q}}
\newcommand{\R}{\mathbb{R}}
\newcommand{\F}{\mathbb{F}}
\newcommand{\CC}{\mathcal{C}}
\newcommand{\curve}{\ensuremath{\CC}}
\newcommand{\defpol}{\ensuremath{f}}
\newcommand{\jac}{\ensuremath{J(\curve)}}
\newcommand{\ideal}[1]{\langle #1 \rangle}

\newcommand{\frobp}{\ensuremath{\sigma}}

\newcommand{\funcring}{\ensuremath{\mathcal{A}}}

\newcommand{\defpollift}{\ensuremath{\overline{\defpol}}}

\newcommand{\frobliftg}{\ensuremath{\mathcal{F}}}
\newcommand{\frobklift}[1]{\ensuremath{\Sigma^{#1}}}
\newcommand{\frobkliftg}[1]{\ensuremath{\mathcal{F}^{#1}}}
\newcommand{\funcringdagger}{\ensuremath{\funcring^{\dagger}}}
\newcommand{\HMW}[1]{\ensuremath{H_{MW}^{#1}(\curve, \Q_{q})}}
\newcommand{\HMWT}[1]{\ensuremath{H_{MW}^{#1}(\curvemw, \Q_{q})}}
\newcommand{\curvemw}{\ensuremath{\widetilde{\curve}}}
\newcommand{\padicprecini}{\ensuremath{N_0}}
\newcommand{\padicprecint}{\ensuremath{N_1}}
\newcommand{\padicprec}{\ensuremath{N}}
\newcommand{\yadicprec}{\ensuremath{\mu}}
\newcommand{\oh}[1]{\ensuremath{\widetilde{O}\left(#1\right)}}

\newcommand{\INDSTATE}{\STATE\hspace{-\algorithmicindent}}

\begin{document}
\title[Point Counting for Cyclic Covers of the Projective Line]{A Point Counting Algorithm for Cyclic Covers of the Projective Line}

\author{C\'{e}cile Gon\c{c}alves}
\address{Laboratoire d'Informatique de l'École polytechnique\\
1 rue Honoré d'Estienne d'Orves\\
Bâtiment Alan Turing\\
Campus de l'École Polytechnique\\
91120 Palaiseau\\
France}

\email{goncalves@lix.polytechnique.fr}
\thanks{The author was supported in part by DGA (D\'el\'egation G\'en\'erale de l'Armement), France.}

\date{}


\begin{abstract}
We present a Kedlaya-style point counting algorithm for cyclic covers $y^r = \defpol(x)$ over a finite field $\F_{p^n}$ with $p$ not dividing $r$, and $r$ and $\deg{\defpol}$ not necessarily coprime.
This algorithm generalizes the Gaudry--G\"urel algorithm for superelliptic curves to a more general class of curves, and has essentially the same complexity. 
Our practical improvements include a simplified algorithm exploiting the automorphism of $\curve$, refined bounds on the $p$-adic precision, and 
an alternative pseudo-basis for the Monsky--Washnitzer cohomology which leads to an integral matrix when $p \geq 2r$.
Each of these improvements can also be applied to the original Gaudry--G\"urel algorithm.
We include some experimental results, applying our algorithm to compute Weil polynomials of some large
 genus cyclic covers.
\end{abstract}

\maketitle

\section{Introduction}
A cyclic cover of the projective line is a nonsingular projective curve $\curve$ defined over
$\F_{p^n}$ by the affine plane model 
$$\curve: y^r = \defpol(x),$$
where $\defpol$ is a monic, squarefree degree $d$ polynomial and $p$ does not divide $r$.
Counting points on, and more generally determining the zeta function of a cyclic cover is an interesting problem with many applications in number theory. 

A lot of work has been done in point counting during 
the last three decades, providing efficient point counting algorithms for elliptic curves ($r = 2$ and $d = 3$). 
The first deterministic polynomial time algorithm for counting points of elliptic curves was the $\ell$-adic algorithm of Schoof \cite{schoof}. 
It was improved by Atkin and Elkies to give the famous SEA algorithm (see \cite{schoofSurvey} for the details).
Pila proposed a generalization of this approach to general abelian varieties \cite{pila, pilaThesis}, and
Schoof-style algorithms have been implemented with success for hyperelliptic curves of genus 2 ($r = 2, d \leq 6$) \cite{gaudrySchost, gaudryKohelSmith}, 
but there seems to be no hope of a practical $\ell$-adic point counting algorithm for $d > 6$ or $r > 2$, since these algorithms are exponential in the genus.
Other efficient point counting algorithms using canonical lifts or the AGM also seem limited to genus 1 and 2, and they are also exponential in $\log p$ \cite{satoh, gaudryHarley, mestre, mestre2002}.

In 2001, Kedlaya \cite{kedlaya_hyperelliptic} published a point counting algorithm for odd degree hyperelliptic curves ($r = 2, d = 2g+1$)
over finite fields of small characteristic $p$. 
This algorithm uses a lift of the Frobenius on Monsky--Washnitzer cohomology (see \cite{monsky} for details) 
and is polynomial in the genus and the field degree, but linear in $p$ (and hence exponential in $\log p$). This complexity was improved by Harvey \cite{harvey} in larger 
characteristic reducing the dependence to $\sqrt{p}$.
Kedlaya's algorithm has been extended to more general classes of curves including superelliptic curves \cite{gaudryGurel_superelliptic}, 
$\curve_{a, b}$ curves 
\cite{denefVercauteren_cab} and nondegenerate curves \cite{castryckDenefVercauteren}. In the latter two cases, the algorithm is 
slower than for superelliptic 
curves of the same genus, because the curves involved are too much general to find a convenient basis of cohomology. 
This problem can be solved using deformation theory \cite{castryckHubrechtsVercauteren_cab} which reduces the computation of the Weil polynomial 
of a $\curve_{a, b}$ curve into the computation of the Weil polynomial of a superelliptic curve. This algorithm has been extended to hypersurfaces \cite{pancratzTuitman}.
Minzlaff \cite{minzlaff} improved the complexity of the Gaudry--G\"urel algorithm for superelliptic curves applying the improvements of Harvey.
Tuitman \cite{tuimanCyclicCover} has recently proposed a Kedlaya-style algorithm for general covers of the projective line.
An alternate approach using Serre duality \cite{besserEscirvaDeJeu} is available for smooth curves over finite fields and appears promising for the general case.

In this paper, we present a Kedlaya--style algorithm for cyclic covers of the projective line which runs in 
$$\oh{pn^3d^4r^3 + n^2 r d ^{\nu+1}\!\!\left(\sum_{i = 1}^sc_i^{\nu}\right)\!\!}$$ 
elementary operations, where the permutation $j \mapsto p^nj \bmod r$ of $\{ 1, \cdots, r-1\}$ is a product of $s$ cycles of lengths 
$c_1, c_2, \cdots, c_s$; $\nu$ is the exponent in the complexity of 
matrix multiplication ($2<\nu<3$), and $\widetilde{O}$ is the Soft-Oh notation which ignores the logarithmic factors. 
Note that in the best case, which is precisely the case where the $r$-th roots of unity are contained in $\F_{p^n}$, we have $s = r-1$ and $c_i = 1$ for $1 \leq i \leq s$, which leads to a complexity in $\oh{pn^3d^4r^3 }$ elementary operations.
In the worst case, $s = 1$ and $c_1 = r-1$, which leads to a complexity in $\oh{pn^3d^4r^3 + n^2r^{\nu+1}d^{\nu+1}}$ elementary operations.
Note also that this approach is compatible with Harvey's improvements so we can reduce the dependency in $p$ to $\sqrt{p}$ in larger characteristic.

This algorithm generalizes the Gaudry--G\"urel algorithm \cite{gaudryGurel_superelliptic} from superelliptic curves to general cyclic covers,
 following the approach of Harrison \cite{harrison} for even-degree hyperelliptic curves.
While our algorithm has essentially the same complexity as the Gaudry--G\"urel algorithm, we offer practical improvements including 
\begin{itemize}
 \item a simplified algorithm exploiting the automorphism of $\curve$;
 \item refined bounds on the $p$-adic precision; and
 \item an alternative pseudo-basis for the Monsky--Washnitzer cohomology which leads to an integral matrix when $p \geq 2r$.
\end{itemize}
Each of these improvements can also be applied to the original Gaudry--G\"urel algorithm.

The paper is organized as follows: Section \ref{sec:covers} recalls the definition of cyclic covers with their main properties, 
Section \ref{sec:adaptation} describes Monsky--Washnitzer cohomology for cyclic covers and the action of Frobenius.  
Section \ref{sec:algo} gives a summary of our algorithm; in Section \ref{sec:complexity} we analyze the 
complexity of our algorithm. Section \ref{sec:precision} proves the precision bounds to which we have to perform the computations 
in order to have an exact result (these bounds also apply to the Gaudry--G\"urel algorithm); 
in Section \ref{sec:basis} we study the use of another pseudo-basis which leads to a matrix with integral coefficients when $p \geq 2r$; the use of this basis slightly 
accelerates the computations.  To conclude, Section 
\ref{sec:an} gives some numerical experiments. 

\section{Cyclic covers of the projective line}\label{sec:covers}
\begin{definition}
A cyclic cover of the projective line is a nonsingular projective curve $\curve$ defined over
$\F_q$ by the affine plane model
\begin{equation}\label{equation:curve}
 \curve: y^r = \defpol(x),
\end{equation}
where $\defpol$ is a monic, squarefree degree $d$ polynomial over $\F_q$ and the characteristic $p$ of $\F_q$ does not divide $r$.
\end{definition}
Let $$\delta := \gcd(r,d).$$
A cyclic cover $\curve : y^r = f(x)$  embeds naturally in the weighted projective space 
$\mathbb{P}\left(\frac{r}{\delta},\frac{d}{\delta},1\right)$ (see \cite{reid} for details on weighted projective spaces), 
where it is a nonsingular curve with $\delta$ points at infinity.
The genus of the curve is
\begin{equation}\label{eq:genus}
 g = \frac{(r-1)(d-1)}{2} - \frac{\delta - 1}{2}.
\end{equation}
It is also equipped with an automorphism of order $r$ defined by $$\rho_r: (x, y)  \longmapsto  (x, \zeta_r y)$$ where $\zeta_r$ is a primitive $r$-th root of unity in $\overline{\F}_q$.

\begin{remark}
When $r$ and $d$ are coprime, then $\delta = 1$ and $\curve$ is superelliptic. 
\end{remark}

\begin{definition}\label{def:weil}
Let $\curve$ be a genus $g$ curve defined over $\F_q$. 
The Weil polynomial $P$ of $\curve$ is the characteristic polynomial of the $q$-th power Frobenius acting on the Jacobian $\jac$ of $\curve$ 
and has the form
$$ P(t) = t^{2g} + a_1 t^{2g-1} + \cdots + a_{g-1} t^{g+1} + a_g t^g + q a_{g-1} t^{g-1} + \cdots + q^{g-1} a_1 t + q^g,$$
with $|a_i| \leq \binom{2g}{i}q^{i/2}$.
We call the coefficients $(a_1,\ldots,a_g)$ the Weil coefficients of $\curve$.
\end{definition}

The aim of any point counting algorithm is to compute the Weil polynomial of the given curve. The Weil polynomial determines the 
cardinality of the Jacobian, since $\# \jac (\F_{q}) = P(1)$. Since the Weil polynomial is the reciprocal polynomial of 
the numerator of the zeta function, it also determines $\# \curve(\F_{q^k})$ (and $\# \jac (F_{q^k})$) for all $k > 0$.

\section{Monsky--Washnitzer cohomology for cyclic covers and the action of Frobenius}\label{sec:adaptation}
Let $\curve: y^r = \defpol(x)$ be a cyclic cover of the projective line of genus $g$ over $\F_q$, with $q= p^n$ and let $d$ denote the degree of $\defpol$.

\vspace{1ex}
We focus on the case where $r$ and $d$ are not coprime, and $\defpol$ has no root in $\F_q$. 
This is because if $r$ and $d$ are coprime, then we can simply apply the Gaudry--G\"urel algorithm 
to $\curve$. 
If $\defpol$ has a root $\alpha$ in $\F_q$, then we can immediately reduce to the case where $r$ and $d$ are coprime. Indeed, let $f_1$ be such that $\defpol(x) = (x - \alpha) f_1(x)$,  
$h(x) = f_1(x + \alpha)$ and $\curve'$ be the superelliptic curve defined over $\F_q$ by
$$\curve': y_1^r = x_1^{d-1} h \left( \frac{1}{x_1} \right).$$

Then the $\F_q$--isomorphism from $\curve'$ to $\curve$ defined by
$$ (x_1, y_1) \longmapsto \left(\frac{1}{x_1 - \alpha}, \frac{y_1}{(x_1 - \alpha)^{d/r}}\right)$$
allows us to apply the Gaudry--G\"urel algorithm to $\curve'$ to compute $P(t)$ 
(the algorithm we describe below reduces to Gaudry--Gürel in the superelliptic case).
Note that in general, $\defpol$ has no root in $\F_q$ so we can't use this trick.

\vspace{1ex}
Recall that $\curve$ has $\delta$ points at infinity. Their coordinates in $\mathbb{P}(\frac{r}{\delta}, \frac{d}{\delta} , 1)$ are
\begin{equation}\label{eq:ptsatinfinity}
 P_{\infty, k} = [1:\zeta_r^k:0] \; \mbox{ for } 1 \leq k \leq \delta,
\end{equation}
where $\zeta_r$ is a primitive $r$-th root of unity over $\overline{\F}_q$.

\vspace{1ex}
In any Kedlaya-style algorithm, we compute in the ring $\Z_q$, which is the ring of integers of $\Q_q$, an unramified extension of $\Q_p$ of degree $n$. 
We have $$\Z_q \cong \Z_p[x] / \ideal{Q(x)}, $$ where $Q$ is an arbitrary lift to $\Z_p$ of a defining polynomial of $\F_q$ over $\F_p$. 
The Galois group of $\Q_q$ over $\Q_p$ is cyclic; its generator $\frobp$ reduces modulo $p$ to the $p$-th power Frobenius automorphism of $\F_q$.

\vspace{1ex}
The aim of our algorithm is to compute the action of Frobenius on the first
Monsky--Washnitzer cohomology group of $\curve$, denoted by 
$\HMW{1}$, and defined by
$$ \HMW{1} = \left( (\funcringdagger dx + \funcringdagger dy) /d \funcringdagger \right) \otimes_{\Z_q} \Q_q,$$
where $\funcringdagger = \Z_q^{\dagger}[[x, y]] / \ideal{y^r - \defpollift(x)}$ with $\defpollift$ an arbitrary lift of degree $d = \deg(\defpol)$ of $\defpol$ to $\Z_q$ and $\Z_q^{\dagger}[[x, y]]$ is the ring of overconvergent series over $\Z_q$, that is the ring of power 
series $\sum_{i, j}{a_{i, j} x^i y^j}$ whose radius of convergence is greater than one. This means that the $p$-adic valuations of the coefficients grows up at least linearly, that is
$$ \exists \alpha, \beta \in \R , \alpha > 0 \; \mbox{ such that } \; \mbox{ord}_p(a_{i, j}) \geq \alpha \abs{i + j} + \beta, \; \forall i, j .$$
Taking a lowbrow point of view, $\HMW{1}$ consists of differential forms of $\curve$ lifted to $\Z_q$
modulo the relations coming from the equation of $\curve$ and the fact that ${d \phi\equiv0}$ for all rational functions $\phi$.
See \cite{monsky} for more detailed explanations of $\HMW{1}$. 

The first thing to do is to determine a basis of $\HMW{1}$. In order to get a basis which is easily described and convenient to compute with, we remove from $\curve$ the ramification points 
of the projection map on the $x$-axis $\pi: \curve \rightarrow \mathbb{P}^1$ defined by ${\pi: [x:y:z] \mapsto [x:z]}$. 
Let $\curvemw$ be the curve corresponding to $\curve$ without the $\F_q$-rational divisor formed by the $d$ points on the $x$-axis and the $\delta$ points at infinity 
(which map to a single rational point of $\mathbb{P}^1$ under $\pi$): 
$$\curvemw = \curve \; \setminus \; \bigg( \Big\{ (\alpha, 0) \in \curve(\overline{\F}_q) \; | \; f(\alpha) = 0 \Big\} \cup \Big\{ P_{\infty, k}\; | \; k \in [1, \delta] \Big\}\bigg).$$
The elements of $\HMWT{1}$ are the $\sum_{i \in \Z, j \in \Z}{(a_{i, j} dx + b_{i, j} dy) x^i y^j}$.
Using the equation of the curve and the relation 
\begin{equation}\label{eqn:curvederived}
 r y^{r-1} dy = \defpollift ' (x) dx,
\end{equation}
it follows that the elements of $\HMWT{1}$ are represented by series in the form $\displaystyle{\sum_{0 \leq i <d, j \in \Z}{a_{i, j} x^i y^j dx}.}$
As Equation \eqref{eqn:curvederived} gives $r x^iy^{j + r-1} dy = x^i y^j\defpollift ' (x) dx$ for all $i, j \in \Z$, and using the fact that 
${d \left(- \frac{r}{r-j} x^i y^{r-j}\right)} = 0$ for $j > r$, we get
\begin{equation}\label{eqn:firstred}
 x^i \defpollift'(x) \frac{dx}{y^j} = \frac{ri}{j-r} x^{i-1} \frac{dx}{y^{j-r}} \; \mbox{ for } j > r.
\end{equation}
If $j \geq 0$, then the relation $y^r = \defpollift(x)$ implies
that the elements of $\HMWT{1}$ are represented by sums in the form $\sum_{0 \leq i <2d, 1 \leq j \leq r }{a_{i, j} x^i \frac{dx}{y^j}}$.
Using ${d (r x^i y^{r-j}) = 0},$ we get
\begin{equation}\label{eqn:secondred}
  \left( ri x^{i-1} \defpollift(x) + (r-j)x^i\defpollift'(x) \right) \frac{dx}{y^j} = 0 \; \mbox{ for } 1 \leq j \leq r \mbox{ and } i \geq 0.
\end{equation}
As the degree in $x$ of the expression above is $d + i-1$, it follows that 
$$\widetilde{B} = \left\{ x^i \frac{dx}{y^j} \; | \; i \in [0, d-2], j \in [1, r] \right\}$$
is a basis of $\HMWT{1}$.

Proposition 6.1 of \cite{monsky} shows that working with $\curvemw$ instead of $\curve$ enlarges the dimension of the first Monsky--Washnitzer cohomology group:
  \[
      \dim\left(\HMWT{1} \right)= 2g + d + \delta - 1 = \dim\left(\HMW{1} \right) + d + \delta - 1.
  \]

The space $\HMWT{1}$ decomposes into a direct sum of $r$ eigenspaces under the action of the automorphism $\rho_r:(x,y)\mapsto(x,\zeta_ry)$ of $\curve$: 
\begin{itemize}
 \item $\HMWT{1}^1$, which corresponds to the fixed points of $\rho_r$ and has dimension $d$ with basis
 $\left\{ x^i \frac{dx}{y^r} \; | \; i \in [0, d-1] \right\}.$
This subspace comes from the $d$ points removed from the affine part of $\curve$
on the $x$-axis.
 \item the $r-1$ spaces $\HMWT{1}^{\zeta_r^{-j}}$ of dimension $d-1$, for $1 \leq j < r$, with basis 
$$\left\{ x^i \frac{dx}{y^j} \; | \; i \in [0, d-2]\right\}.$$
\end{itemize}
Observe that $\rho_r \left( x^i \frac{dx}{y^j} \right) = \zeta_r^{-j} x^i \frac{dx}{y^j}$, so $\rho_r$ has eigenvalue $\zeta_r^{-j}$ on $\HMWT{1}^{\zeta_r^{-j}}$ for any $1 \leq j \leq r$.

Let $i: \HMW{1} \hookrightarrow \HMWT{1}$ be the embedding from $\HMW{1}$ to $\HMWT{1}$. As $i ( \HMW{1}) \cap \HMWT{1} = \{ 0 \}$, the space $\HMW{1}$ is contained in the direct sum 
of the $\HMWT{1}^{\zeta_r^j}$, for $1 \leq j < r$ which has dimension $2g + \delta - 1$ with basis 
$$B = \left\{ x^i \frac{dx}{y^j} \; | \; i \in [0, d-2] \; , j \in [1,r-1] \right\}.$$

Clearly $B$ does not correspond to a basis for $\HMW{1}$ (since it contains $\delta-1$ too many vectors). 
The space generated by $B$ decomposes into a direct sum of two subspaces stable under the action of Frobenius.
Let ${c: \ideal{B} \rightarrow \HMW{1}}$ be the map sending an element to its representative in $\HMW{1}$.
As ${\HMWT{1}}$ contains $\HMW{1}$ and ${i(\HMW{1}) \cap \HMWT{1}^1 = \{ 0 \}}$ (where $i$ is the embedding from $\HMW{1}$ to $\HMWT{1}$), 
it follows that $c$ is surjective. So we have the decomposition 
$$ \ideal{B} = \HMW{1} \oplus \ker(c).$$

So we will proceed into two steps: we first compute the action of Frobenius on the space generated by $B$, and then we remove an extra factor from its characteristic 
polynomial which corresponds to the action of Frobenius on the $(\delta-1)$-dimensional subspace $\ker(c)$. 
Theorem \ref{theo:removedFactor} describes this extra factor.
This is similar to Harrison's approach in extending Kedlaya's algorithm to hyperelliptic curves with even degree \cite{harrison}.

\begin{remark}\label{rem:block}
The $p$-th power Frobenius maps the space $\HMWT{1}^{\zeta_r^j}$ to $\HMWT{1}^{\zeta_r^{\ell}}$ for $1 \leq j < r$ where ${\ell = jp \bmod r}$. Thus, the matrix of the $p$-th power Frobenius 
map acting on $B$ is a block matrix with exactly $r-1$ non-zero blocks of size $d-1$. 
There is exactly one non-zero block on each row partition and each column partition.
\end{remark}

\vspace{1ex}
We begin by computing the action of Frobenius on the elements of $B$. Let $$ \tau = y^{-r}.$$
We can lift the $p$-th power Frobenius to differential forms by setting 

$$\frobliftg(x) = x^{p}$$
and
$$\frobliftg(y)^r  = \frobliftg(\defpollift(x))$$
so
\begin{align}\label{eqn:froby}
\frobliftg(y) & =  y^{p} \left( 1  + \left( \frobliftg \left(\defpollift(x)\right) - \defpollift(x)^{p} \right){\tau^{p}}\right)^{\frac{1}{r}} \\
 &  = y^p\sum_{i \geq 0}{\binom{1/r}{i} \left(\frobliftg\left(\defpollift(x)\right) - \defpollift(x)^{p}\right)^i \tau^{p i}},\notag
\end{align}
and 
$$ \frobliftg(dx)  = d (\frobliftg(x)) =  p x^{p-1}dx. $$

The $p$-th power Frobenius acts on the basis vectors as
\begin{align}\label{eqn:frobB}
 \frobliftg \left( x^i \frac{dx}{y^j} \right) & =  p x^{p(i+1) - 1} \frobliftg(y)^{-j}{dx} \\
 &=  p x^{p(i+1) - 1} y^{-jp} ( 1 + p E(x)\tau^{p} )^{-j/r}{dx} \notag \\
 & =  x^{p(i+1) - 1} y^{-jp} \sum_{k \geq 0}{\binom{-j/r}{k}p^{k+1}E^k(x) \tau^{pk}} dx \notag\\
 & =  x^{p(i+1) - 1}\sum_{k \geq 0}{\binom{-j/r}{k}p^{k+1}E^k(x) \tau^{pk + a}} \frac{dx}{y^{\ell}} \notag,
\end{align}
where $E(x) = \frac{\frobliftg(\defpollift(x)) - \defpollift(x)^{p}}{p}$ and $jp = ar + \ell$.
{Applying a change of index, we obtain
$$\frobliftg \left( x^i \frac{dx}{y^j} \right) = x^{p(i+1) - 1} \!\sum_{k \geq a, k \equiv a \bmod p}{p^{\frac{k-a}{p} + 1}Q_k(x) \tau^{k}} \frac{dx}{y^{\ell}} \!.$$} 
After normalizing (using the equation of $\curve$ to remove all the terms with degree in $x$ greater than $\deg(\defpollift)$), and using the fact that $Q_a = 1$, we have: 
\begin{equation}\label{eqn:frobdiff}
\frobliftg \left( x^i \frac{dx}{y^j} \right) = \! \sum_{k \geq k_0, k \equiv a \bmod p} \!{p^{\frac{k-a}{p} + 1}R_k(x) \tau^{k}} \frac{dx}{y^{\ell}}, \! \; \mbox{ where } k_0 \geq a - \left\lfloor \frac{p(i+1) - 1}{d} \right\rfloor
\end{equation}
with $\deg(R_k) < d$ for all $k>1$.

We then apply two reduction steps described below, resulting from relations in cohomology, in order to express the image of each basis element as a linear combination of the elements of $B$.

\begin{description}
 \item[Red1] Decrease the degree in $\tau$ by at least one, using the formula
\begin{equation}\label{eqn:Red1}
 R_k(x) \tau^k \frac{dx}{y^{\ell}} = \left( A_k(x) + \frac{r}{r(k-1)+ \ell } B_k'(x) \right) \tau^{k-1} \frac{dx}{y^{\ell}}, 
\end{equation}
where $ R_k = A_k\defpollift + B_k \defpollift'$, coming from Equation \eqref{eqn:firstred}. We apply \eqref{eqn:Red1} as many times as needed.
Note that $A_k$ and $B_k$ exist because $\defpol$, and therefore $\defpollift$ are squarefree. 
 \item[Red2] Given an expression of the form $S(x) \frac{dx}{y^{\ell}}$ with $S$ of degree $m$, use Equation \eqref{eqn:secondred} to decrease by at least one the degree of $S$ to at most $d-2$ by subtracting some multiples of 
\begin{equation}\label{eqn:Red2}
 r(i - d+1)x^{i - d}\defpollift(x) + (r - \ell )x^{i - d+1} \defpollift'(x) \; \mbox{ for } d-1 \leq i \leq m = \deg(S)
\end{equation}
from $S$. We apply \eqref{eqn:Red2} as many times as needed. 
\end{description}

\vspace{1ex}
We obtain a matrix $M_{\frobliftg}$, which is the matrix of the $p$-th power Frobenius with respect to $B$. We recover the matrix of the $q$-th power Frobenius which is 
\begin{equation}\label{eqn:matrixFrob}
 M = M_{\frobliftg} \cdot M_{\frobliftg} ^{\sigma}  \cdots M_{\frobliftg}^{\frobp^{n-1}},
\end{equation}
where $\frobp$ is the $p$-th power Frobenius on $\Q_q$ and $M_{\frobliftg}^{\sigma}$ is the matrix obtained by applying $\frobp$ to the coefficients of $M_{\frobliftg}$. 
Note that we can use the special block structure of $M_{\frobliftg}$ to speed up the computation of $M$, which is a matrix composed of $(d-1)\times(d-1)$ blocks as well.

\section{Adaptation of the Gaudry--Gürel algorithm to general cyclic covers}\label{sec:algo}
We want to compute the Weil polynomial of $\curve$. Once we have computed the characteristic polynomial $\chi_M$, we need to remove an extra factor.
\begin{theorem}\label{theo:removedFactor}
The Weil polynomial $P$ of $\curve$ is
$$P(t) = \frac{\chi_{M}(t)}{U(t)}, $$
where $\chi_{M}(t)$ is the characteristic polynomial of  the matrix $M$ corresponding to the action of the $q$-th power Frobenius with respect to $B$ and $$U(t) = \displaystyle{\prod_{i | \delta \;, \; i > 1 }(t^{k_i} - q)^{\frac{\varphi(i)}{k_i}}},$$ where 
$k_i$ is the order of $q$ in $\Z / \varphi(i) \Z$ and $\varphi$ is the Euler totient function.
\end{theorem}
\begin{proof}
We follow the approach of \cite[Lemma 3.1]{harrison}.

Let $\alpha_1, \cdots, \alpha_{2g}$ denote the roots of $P$ and $S_e = \alpha_1^e + \alpha_2^e + \cdots + \alpha_{2g}^e$.
Let $R_e$ denote the number of roots of $f$ in $\F_{q^e}$ and $I_e$ the number of $\F_{q^e}$-points at infinity.
Finally, let
$$ L = \pi(\curvemw), $$
where $\pi : \curve \rightarrow \mathbb{P}^1$ is the projection on the $x$-axis.

For every $e > 0$, we have
  $$\# \curve (\F_{q^e}) = (q^e + 1 - S_e), $$
  $$\# \curvemw (\F_{q^e}) = (q^e + 1 - S_e) - R_e - I_e, $$
and
 $$\# L (\F_{q^e}) = (q^e + 1) - R_e - 1 .$$

The Lefschetz trace formula for $\curvemw$ says that for each $e > 0$, 
\begin{align*}
 (q^e + 1 - S_e) - R_e - I_e   =  & \; \mbox{Tr}( (q\frobliftg^{-1})^e | \HMWT{0} )\\
 & - \mbox{Tr}( (q\frobliftg^{-1})^e | \HMWT{1}^1 )\\
  &- \sum_{j \neq 0} \mbox{Tr}( (q\frobliftg^{-1})^e | \HMWT{1}^{\zeta_r^{-j}} )
\end{align*}
while the trace formula for $L$ says that for each $e > 0$,
$$(q^e + 1) - R_e - 1   = \mbox{Tr}( (q\frobliftg^{-1})^e | \HMWT{0}) - \mbox{Tr}( (q\frobliftg^{-1})^e | \HMWT{1}^1 ).$$
Subtracting the first equation from the second, we get
$$ \sum_{j \neq 0} \mbox{Tr}( (q\frobliftg^{-1})^e | \HMWT{1}^{\zeta_r^{-j}} )   =   S_e + I_e - 1 .$$
The sum of the $e$-th powers of the eigenvalues of $q\frobliftg^{-1}$ on $\oplus_j \HMWT{1}^{\zeta_r^{-j}}$ is $S_e + I_e - 1$.

Let $$V(t) =\prod_{i | \delta, i > 1}(t^{k_i} - 1)^{\frac{\varphi(i)}{k_i}}, $$
where $k_i$ is the order of $q$ in $\Z / \varphi(i) \Z$ and $\varphi$ is the Euler totient function.
Then the $e$-th power sum of the roots of $V(t)$ is $I_e - 1$ for each $e>0$, so the $e$-th power sum of the roots of $V(t)P(t)$ is $S_e + I_e - 1$, or each $e>0$.
Then $$P(t)V(t) = \chi_{q\frobliftg^{-1}}(t),$$
and hence $$P(t)U(t) = \chi_M(t).$$
\end{proof}

\begin{remark}
Note that in the superelliptic case, the fact that $\delta = 1$ tells us that $B$ corresponds to a basis of $\HMW{1}$.
So in the superelliptic case we do not have to remove an extra factor from the characteristic polynomial of the Frobenius map acting on $B$: the characteristic polynomial is the Weil polynomial of $\curve$. 
\end{remark}

In practice, we want to compute the Weil polynomial of $\curve$ over $\Z_q$, but $p$-adic numbers are infinite series so algorithmically we are 
forced to work with finite approximations. In practice, we work up 
to a certain precision $\padicprec$: that is, modulo $p^N \Z_q$ for some suitably 
large value of $\padicprec$.

The Weil bounds (see Definition \ref{def:weil}) tell us that if we do the computations to sufficient precision then we can recover 
the Weil polynomial exactly. As we lose some digits of precision during the reduction steps, we have to enlarge these bounds. 
Theorem \ref{theo:bounds} states the precision bounds and its proof can be found in \S \ref{sec:precision}.

\begin{theorem}\label{theo:bounds}
In order to compute the Weil polynomial of $\curve$ exactly, we have to do the computations with the basis $B$ up to precision
$$\padicprec = \min_{ n \in \N} \left\{n - \bigg\lfloor \log_p \Big( p(rn-1) - r \Big)  \bigg\rfloor  \geq \padicprecini + \bigg\lfloor \log_{p}\Big(\frac{dp(r-1)+r}{\delta}\Big)\bigg\rfloor \right\}$$
where
$$\padicprecini = \left\lceil \log_{p} \left( 2\binom{2g}{g} q^{g/2} \right) \right\rceil.$$
\end{theorem}

We will carry out our computations modulo $p^{\padicprec}$, where $\padicprec$ is defined in Theorem \ref{theo:bounds}.
At the end of the algorithm, we will have determined the Weil polynomial of $\curve$ modulo $p^{\padicprecini}$, which is sufficient to determine it exactly, because of the 
  Weil bounds.

\vspace{1ex}
Algorithm \ref{algo:cover}, \textsc{CyclicCoverWeilPolynomial} computes the Weil polynomial of a cyclic cover $\curve$ defined over $\F_q$ by the equation $y^r = \defpol(x)$. 
Note that Steps 1 to 5 of Algorithm \ref{algo:cover} reduces to Gaudry and G\"urel's algorithm.

Step 1 computes the precision $\padicprec$ stated by Theorem \ref{theo:bounds}.
In Step 2, we compute $\frobliftg(y)^{-1} \bmod p^{\padicprec}$ using Equation \eqref{eqn:froby}. 
Indeed, if ${R = 1  + (\frobliftg(\defpollift(x)) - \defpollift(x)^{p}){\tau^{p}}}$, then ${\frobliftg(y)^{-1} = y^{-p} R ^ {-\frac{-1}{r}}}$, where $\defpollift$ is an arbitrary degree 
$d = \deg(\defpol)$ lift of $\defpol$ to $\Z_q$ and $\tau = y^{-r}$.
Note that we use a Newton iteration to compute the inverse of the $r$-th root of $R$ up to precision $\padicprec$.
In Step 3, we compute the action of Frobenius on the vectors of $B$ given by Equation \eqref{eqn:frobB}. 
We then apply Algorithm \ref{algo:reductionSuperelliptic}, \textsc{ReductionCohomology}, to reduce a differential form back to a linear combination of vectors of $B$ 
using the reduction rules Red1 and Red2 described above in Equations \eqref{eqn:Red1} and \eqref{eqn:Red2}.
The result of Step 3 is the matrix $M_{\frobliftg}$ of the $p$-th power Frobenius map acting on $B$. The ordering of the indices $\mathcal{J}$ is chosen in such a way that the $d-1$ by $d-1$ 
blocks are grouped into larger blocks reflecting the cyclic action of multiplication by $q$ modulo $r$. These larger blocks form a block diagonal matrix: for each cycle of length $c$, there is a block of size $c(d-1)$.
In Step 4, we compute the matrix $M$ of the $q$-th power Frobenius map acting on $B$ from $M_{\frobliftg}$ using Equation \eqref{eqn:matrixFrob} and using the block structure of $M_{\frobliftg}$. Note that the resulting matrix $M$ has the same block structure. 
We then compute its characteristic polynomial up to precision $\padicprecini$. Again, we use the block structure of $M$.
Finally, in Step 5 we compute the extra factor $U$ using Theorem \ref{theo:removedFactor} and we then return 
the Weil polynomial of $\curve$.

\begin{algorithm}
\floatname{algorithm}{Algorithm}
\caption{\textsc{CyclicCoverWeilPolynomial}}
\label{algo:cover}
\begin{algorithmic}
\REQUIRE Cyclic cover $\curve$ of genus $g$ over $\F_q$, with $q = p^n$,  defined by the equation $y^r = \defpol(x)$, 
with $\defpol$ a monic, squarefree degree $d$ polynomial.
\ENSURE The Weil polynomial $P(t)$ of $\curve$.
\INDSTATE \textbf{Step 1: Precision bounds:}\\
\STATE Compute $\delta := \gcd(r, d);$
\STATE $\padicprecini:= \left\lceil \log_{p}\left(2 \binom{2g}{g} q^{g/2} \right)\right\rceil;$ 
\STATE $\padicprec:= \min \left\{ n - \left\lfloor \log_p\left(p(rn-1)-r \right)\right\rfloor  \geq \padicprecini + \left\lfloor \log_{p}\left(\frac{dp(r-1)+r}{\delta} \right)\right\rfloor | {n \in \N} \right\}$;
\INDSTATE \textbf{Step 2: Compute $\frobliftg(y)^{-1} \bmod p^{\padicprec}$:}
\STATE $R:= 1  + (\frobliftg(\defpollift(x)) - \defpollift(x)^{p}){\tau^{p}}$;
\COMMENT{ where $\tau$ is $y^{-r}$ and $\defpollift$ is an arbitrary lift of $\defpol$ to $\Z_q$.}
\STATE $S := R^{-\frac{1}{r}};$ 
\COMMENT{$\frobliftg(y)^{-1} = y^{-p} S$}
\INDSTATE \textbf{Step 3: Action of Frobenius on $B$ : }\\
\COMMENT{where $B = \left\{ x^i \frac{dx}{y^j} \; | \; i \in [0, d-2] \; , j \in [1,r-1] \right\}$.}
\STATE Compute $\mathcal{J}$ as the sequence $[1, r-1]$ sorted in cycles under the action of multiplication by $q$ mod $r$.
\COMMENT {We have $J[i+1] = qJ[i] \bmod r$.}
\FOR {$j$ in $\mathcal{J}$}
\FOR {$i=0$ to $d-2$}
   \STATE  $\omega:= px^{p(i+1) -1} S^j \tau^{\frac{jp - \ell}{r}} \frac{dx}{y^{\ell}};$
    \COMMENT{ where $\ell = jp \bmod r$.} \\
   \COMMENT{ $\omega = \frobliftg \left( x^i \frac{dx}{y^j}\right)$ has the form $\sum_{0 \leq k \leq \yadicprec}{R_k(x) \tau^k \frac{dx}{y^{\ell}}},$ where $\yadicprec = p\padicprec-1$.}
   \STATE  Red:= \textsc{ReductionCohomology}$(\curve, \padicprec, \omega);$ \\
\COMMENT{Then, fill the matrix}
    \FOR {$k = 0$ to $d-2$}
      \STATE $M_{\frobliftg}[(d-1)(j-1) + i + 1][(d-1)(\ell -1) + k + 1]:= \mbox{Coeff}(k, \mbox{Red});$\\
    \ENDFOR
\ENDFOR
\ENDFOR
\INDSTATE \textbf{Step 4: Compute the characteristic polynomial:} 
\STATE $M:= M_{\frobliftg}.M_{\frobliftg} ^{\sigma}  \ldots M_{\frobliftg}^{\frobp^{n-1}}$;
\COMMENT{Use the block structure of $M_{\frobliftg}$ to speed up the computation of $M$}
\STATE $\chi(t):= \chi_{M}(t) \mod \padicprecini$;
\COMMENT{Use the block structure of $M$ to speed up the computation of $\chi$}
\INDSTATE \textbf{Step 5: Remove the extra factor:}
\STATE $U(t):=\prod_{i | \delta \;, \; i > 1 }(t^{k_i} - q)^{\frac{\varphi(i)}{k_i}},$ where $k_i := \min{\{n \in N : \varphi(i) \; | \;q^{n}-1 \}};$
\RETURN $P(t):= \frac{\chi(t)}{U(t)};$
\end{algorithmic}
\end{algorithm}

\begin{algorithm}
\floatname{algorithm}{Algorithm}
\caption{\textsc{ReductionCohomology}}
\label{algo:reductionSuperelliptic}
\begin{algorithmic}
\REQUIRE A lift of a cyclic cover $\curve : y^r = \defpol(x)$ to $\Z_q$ up to precision $\padicprec$, with $q = p^n$, and $\defpol$ a monic, squarefree degree $d$ polynomial,
 precision $\padicprec$,
 differential form $\omega = \sum_{0 \leq k \leq \yadicprec}{R_k(x) \tau^k \frac{dx}{y^{\ell}}}$.
\ENSURE A differential form $T(x) \frac{dx}{y^{\ell}}$ equivalent to $\omega$, with $\deg(T) \leq d-2$.
\INDSTATE \textbf{Step 1: Reduce degree in $\tau$, that is, transform $\omega$ to $S(x) \frac{dx}{y^{\ell}}$:} \\
\COMMENT{ At each iteration, we reduce the degree in $\tau$ by at least one.}
\FOR {$k=\yadicprec$ to $1$}
   \STATE  Compute $A_k$ and $B_k$ such that $R_k = A_k \defpollift + B_k \defpollift'$, using the extended Euclidean
algorithm.
\COMMENT{$\defpollift$ is an arbitrary lift of $\defpol$ to $\Z_q$.}
   \STATE  Replace the term $R_k(x) \tau^k\frac{dx}{y^{\ell}}$ in $\omega$ with $\left( A_k(x) + \frac{r}{r(k-1)+ \ell }B_k'(x)\right) \tau^{k-1}\frac{dx}{y^{\ell}}$.\\
\ENDFOR
\INDSTATE \textbf{Step 2: Reduce the degree in $x$, that is, transform $S(x) \frac{dx}{y^{\ell}}$ to $T(x) \frac{dx}{y^{\ell}}$, with $\deg(T) \leq d-2$:} \\
\STATE $T:= S;$
\STATE $m:= \deg(T);$
\WHILE{ $m \geq d-1$}
  \STATE $\widetilde{T}:=  r(m - d+1)x^{m - d}\defpollift(x) + (r - \ell )x^{m - d+1} \defpollift'(x);$ \\
  \STATE $\widetilde{T}:= \frac{LC(T)}{LC(\widetilde{T})} \widetilde{T};$ \COMMENT{LC is the Leading Coefficient} \\
  \STATE $T:= T - \widetilde{T};$\\
  \STATE $m:= \deg(T);$\\
\ENDWHILE
\RETURN $T(x)\frac{dx}{y^{\ell}};$
\end{algorithmic}
\end{algorithm}

\section{Complexity}\label{sec:complexity}
In this section, we describe the space and time complexity of Algorithm \ref{algo:cover} as a function of the parameters $p$, $n$, $r$ and $d$ 
and we show that this complexity is linear in $p$ and polynomial 
in $n$, $r$ and $d$ (in particular, it is polynomial in the genus). We use the Soft-oh notation so that the logarithmic terms are not taken into account.
We let $2 < \nu < 3$ be the exponent such that the complexity of multiplying two square matrices of size $k$ over a ring $\mathcal{R}$ is $\oh{k^{\nu}}$ operations in $\mathcal{R}$ (using the Coppersmith--Winograd algorithm, for example).
We let $s$ be such that the permutation ${j \mapsto qj \bmod r}$ of $\{ 1, \cdots r-1\}$ is a product of $s$ cycles of lengths $c_1, c_2, \cdots, c_s$.

\begin{prop}
With the notation above, the Weil polynomial of a cyclic cover $\curve: y^r = \defpol(x)$ defined over $\F_{p^n}$, with $\defpol$ of degree $d$ can 
be computed using Algorithm \ref{algo:cover} in time $\oh{pn^3d^4r^3 + n^2 r d ^{\nu+1}\!\!\left(\sum_{i = 1}^sc_i^{\nu}\right)\!}$ 
elementary operations and space $\oh{pn^3d^3 r^2  + n^2 d^3 r \! \left( \sum_{i = 1}^{s}c_i^{2} \right) \!}$ bits of memory.
\end{prop}
\begin{proof}
We first describe the bit size of the different objects.
An element of $\Z_q$ is represented by a polynomial of degree $n-1$ with coefficients in $\Z_p$ truncated to precision $\padicprec$, so it has size $O(n\padicprec \log(p)) = \oh{n \padicprec}$.
An element of $\HMW{1}$ is represented by a degree $\yadicprec = O(p \padicprec)$ polynomial in $\tau$ whose coefficients are polynomials over $\Z_q$ of degree less than $d$, 
so it has size $\oh{pnd\padicprec^2}$. 

\vspace{1ex}
Thanks to Sch\"onhage and Strassen, the multiplication between two integers of bit-size $k$ is $\oh{k}$ elementary operations \cite{gathen}. Hence, the complexity of multiplying two 
elements of $\Z_q$ is $\oh{n\padicprec}$ elementary operations and 
the cost of multiplying two elements of $\HMW{1}$ is $\oh{pnd \padicprec^2}$ elementary operations.

\vspace{1ex}
Recall that we normalize the elements of $\HMW{1}$ at each step using the equation of the curve, so we need to calculate the complexity of the normalization.
This procedure is described in \cite{gaudryGurel_hyperelliptic, gathen}, so we will not go into further detail here. If we want to normalize $Q(x) \tau^k$, then it costs $\oh{\deg(Q)}$ 
operations in 
$\Z_q$ so the complexity of the normalization step is $\oh{n \padicprec \deg(Q)}$ elementary operations.

\vspace{1ex}
We compute the Frobenius substitution on $\Z_q$ by a Newton iteration and H\"orner's method: if $z = \sum_{k = 0}^{n-1}{z_k t^k}$ is an element of 
$\Z_q$ then ${z^{\sigma} = \sum_{k = 0}^{n-1}{z_k t^{\sigma k}}}$ where $t^{\sigma}$ is computed using a Newton iteration.
The complexity of the Newton algorithm is determined by the last iteration, which costs $O(n)$ operations in $\Z_q$, that is $\oh{n^2 \padicprec}$ elementary operations.
H\"orner's method costs $O(n)$ operations in $\Z_q$, that is, $\oh{n^2 \padicprec}$ elementary operations. Hence, we compute the Frobenius substitution on $\Z_q$ in $\oh{n^2 \padicprec}$ 
elementary operations.

\vspace{1ex}
In Step 2 of Algorithm \ref{algo:cover}, we compute the inverse of $\frobliftg(y)$ by a Newton iteration. We first compute $R$, by computing the polynomial 
$\frobliftg(\defpollift(x)) - \defpollift(x)^{p}$ of degree $pd$. We then normalize $R$ which costs $\oh{n \padicprec \deg(R)} = \oh{n \padicprec p d}$ elementary operations. 
We then apply the Newton algorithm to $R$ in order to compute $S$ as the inverse of its $r$-th root. The complexity of the Newton algorithm is 
a constant times the cost of its last 
iteration which consists of some multiplications between two elements of $\HMW{1}$ and a normalization. 
The cost of the Newton iteration is therefore $\oh{pnd \padicprec^2}$ elementary operations.
So the cost of Step 2 is $\oh{pnd \padicprec^2}$ elementary operations.
This step requires $\oh{pnd \padicprec^2}$ bits of memory.

\vspace{1ex}
In Step 3 of Algorithm \ref{algo:cover}, we first compute $S^j$ which costs $\oh{pnd \padicprec^2}$ elementary operations. 
We then apply a normalization to $px^{p(i+1)-1} S^j$ which costs $\oh{pd n\padicprec}$ elementary operations in the worst case. 
Then we perform the reduction steps using Algorithm \ref{algo:reductionSuperelliptic}.
In Step 1 of Algorithm \ref{algo:reductionSuperelliptic}, we do $\yadicprec$ iterations: each time we compute $A_k$ and $B_k$ using the extended Euclidean algorithm,
 which costs $O(d)$ operations in $\Z_q$, that is $\oh{dn \padicprec}$ elementary operations. 
Then we replace the term in $\omega$ of highest degree in $\tau$ by performing $d$ additions 
in $\Z_q$, which costs $\oh{nd\padicprec}$ elementary operations. So the cost of Step 1 of Algorithm \ref{algo:reductionSuperelliptic} is $\oh{\yadicprec d n \padicprec} = \oh{pdn \padicprec^2}$ elementary operations.
During the second Step of Algorithm \ref{algo:reductionSuperelliptic}, we do $\deg(S) \leq dp$ iterations which consists of $d$ operations of elements of $\Z_q$, so the cost of Step 2 in Algorithm 
\ref{algo:reductionSuperelliptic} is $\oh{dp \times d \times n \padicprec} = \oh{pd^2n \padicprec}$ elementary operations and it requires 
$\oh{pnd \padicprec^2} $ bits of memory.

The complexity of reducing a differential form with Algorithm \ref{algo:reductionSuperelliptic} is therefore $\oh{pdn\padicprec(\padicprec + d)}$. As we need to reduce $O(rd)$ differential 
forms, the complexity of Algorithm \ref{algo:reductionSuperelliptic} is 
$\oh{pd^2nr\padicprec(\padicprec + d)}$
elementary operations, and since we reduce differential forms one by one, Algorithm \ref{algo:reductionSuperelliptic} requires $\oh{pnd \padicprec^2} $ bits of memory.
Putting everything together, the complexity for Step 3 of Algorithm \ref{algo:cover} is $\oh{pd^2nr\padicprec(\padicprec + d)}$ elementary operations and 
it requires $\oh{pnd \padicprec^2} $ bits of memory.

\vspace{1ex}
In Step 4 of Algorithm \ref{algo:cover}, we compute the matrix of the $q$-th power Frobenius. 
Recall that $M_{\frobliftg}$ is a block diagonal matrix of $s$ blocks 
matrices $M_{\frobliftg,i}$ of size ${c_i(d-1) \times c_i (d-1)}$ for $1 \leq i \leq s$. Note that $M_{\frobliftg,i}$ is itself a block matrix, composed of $c_i$ blocks of size $(d-1) \times (d-1)$, 
with only one non zero block on each row partition and column partition.
The matrix $M$ is also a block diagonal matrix: its blocks are the norms 
$$M_i = M_{\frobliftg, i} \cdot M_{\frobliftg, i} ^{\sigma}  \cdots M_{\frobliftg, i}^{\frobp^{n-1}}.$$
Each of the $M_{i}$ can be computed using H\"orner's method (and the sub-block structure of $M_{\frobliftg, i}$);
this costs
$\oh{nc_i d^{\nu}}$ operations in $\Z_q$ and requires $\oh{c_i d^2 n \padicprec}$ bits of memory. 
The total cost of computing $M$ is therefore $\oh{n \left( \sum_{i = 1}^{s} {c_i} \right)d^{\nu}n\padicprec}$, that is, $\oh{n^2 r d^{\nu}\padicprec}$ elementary operations (because $\sum_{i = 1}^{s} {c_i} = r$ ) and requires 
$\oh{rd^2 n \padicprec}$ bits of memory.

We then compute the characteristic polynomial of $M$ which is the product of the characteristic polynomials $\chi_{M_i}$ of the $M_i$.
The complexity of computing the characteristic polynomial of a square matrix of size $k$ over a ring $\mathcal{R}$ is $\oh{k^{\nu}}$ operations in $\mathcal{R}$.
Hence, computing $\chi_{M_i}$ costs $\oh{(c_id)^{\nu}}$ operations in $\Z_q$.
So, computing $\chi_M$ costs $\oh{\left( \sum_{i = 1}^{s}{c_i^{\nu}} \right)d^{\nu}n \padicprec}$ elementary operations and 
requires $\oh{\left( \sum_{i = 1}^{s}{c_i^2}\right) d^2 n \padicprec}$ bits of memory.

The global cost of Step 4 is therefore 
$\oh{ \left( n r + \left( \sum_{i = 1}^{s}c_i^{\nu} \right) \right) d^{\nu}  n\padicprec}$ 
elementary operations and 
requires $\oh{\left( \sum_{i = 1}^{s}c_i^2\right) d^2 n \padicprec}$ bits of memory.

\vspace{1ex}
In Step 5 of Algorithm \ref{algo:cover}, we compute the polynomial $U$ of degree $\delta - 1$ over the integers. 
This corresponds to multiplying at most $\delta - 1$ binomials with coefficients no larger that $p^n$.
So this costs $\oh{n\delta}$ elementary operations.
We then divide the polynomial obtained in Step 5 by $U$, so Step 5 costs the equivalent of $O(g)= O(rd)$ operations in $\Z_q$, 
that is $\oh{rd n \padicprec}$ elementary operations and requires $\oh{rd\padicprec}$ bits of memory.

\vspace{1ex}
The total complexity of our algorithm is therefore 
$$ \oh{pnd \padicprec^2 + pd^2nr\padicprec(\padicprec + d) + \left( n r + \left( \sum_{i = 1}^{s}c_i^{\nu} \right) \right) d^{\nu}  n\padicprec + rd n \padicprec}$$
elementary operations and $$\oh{pnd \padicprec^2 + d^2 n \padicprec \! \!\left( \sum_{i = 1}^{s}c_i^{2} \right)\! \!} $$ bits of memory.
Theorem \ref{theo:bounds} tells us that $\padicprec = \oh{ng} = \oh{nrd}$, so the complexity of our 
algorithm is 
$$\oh{pn^3d^4r^3 + n^2 r d ^{\nu+1}\!\!\left(\sum_{i = 1}^sc_i^{\nu}\right)\!\!}$$
 elementary operations and $$\oh{pn^3d^3 r^2  + n^2 d^3 r \! \!\left( \sum_{i = 1}^{s}c_i^{2} \right)\! \!} $$ bits of memory.
\end{proof}

\begin{remark}
Note that if $q \equiv 1 \bmod r$, then $s=r-1$ and $c_i = 1$ for $1 \leq i \leq s$, and hence the complexity of Algorithm \ref{algo:cover} is $\oh{pn^3d^4r^3}$, which is the better case.
\end{remark}

\begin{remark}
We can improve the complexity of this algorithm in larger characteristic to $\oh{\sqrt{p}n^3d^4r^3 + n^2 r d ^{\nu+1}\!\!\left(\sum_{i = 1}^sc_i^{\nu}\right)\!}$, by applying Harvey's improvements \cite{harvey} to Kedlaya's algorithm which 
were extended to superelliptic curves by Minzlaff \cite{minzlaff}. Indeed, our algorithm is entirely compatible with Minzlaff's improvements and we can apply them to our algorithm. 
These improvements consist of two major key points. First, they use a different representation for the images of differential forms under the action of Frobenius: in Kedlaya's algorithm, 
we use an approximation by series whose number of terms is linear in $p$ whereas Harvey's improvements use a different series approximation which does not depend on $p$.
Second, these improvements reduce the complexity of the reduction algorithm, which is the major step, by solving a linear recurrence using the Bostan--Gaudry--Schost algorithm 
\cite{bostanGaudrySchost}.
\end{remark}

\section{Bounds on precision}\label{sec:precision}

In this section, we give a proof of Theorem \ref{theo:bounds}.
In order to compute the Weil polynomial exactly, we need to take sufficient precision. The Weil bounds give us a minimal bound $\padicprecini$, but this bound is not sufficient since the
divisions in the reduction algorithm (Algorithm \ref{algo:reductionSuperelliptic}) induce a loss of precision. Proposition \ref{lemma:bound1} estimates the loss of precision during the first step of Algorithm \ref{algo:reductionSuperelliptic}, 
and Proposition \ref{lemma:bound2} estimates the loss of precision during the second step of Algorithm \ref{algo:reductionSuperelliptic}.

\begin{prop}\label{lemma:bound1}
If $R$ is a polynomial defined over $\Z_q$ with degree less than $d$, then the first step in Algorithm \ref{algo:reductionSuperelliptic} transforms 
$R(x) \tau^k \frac{dx}{y^{\ell}}$, with $k > 1$, into $S(x) \frac{dx}{y^{\ell}}$, where $S$ is a polynomial defined over $\Q_{q}$ with degree less than $d$.
Moreover, the coefficients of $S$ have denominator bounded by $p^{\left\lfloor \log_{p}\left(r(k-1)+ \ell \right)\right\rfloor}$.
\end{prop}

\begin{proof}
We follow the approach of \cite[Lemma 2]{kedlaya_hyperelliptic} and \cite[Lemma 4.3.4]{edixhoven}.

During the first step of Algorithm \ref{algo:reductionSuperelliptic}, we apply \eqref{eqn:Red1} several times. Hence, divisions by 
$r(i-1)+ \ell $ are done, which corresponds to negative powers of $y$ appearing during this step, for each $1 \leq i \leq k$, and positive powers of $p$ may 
occur in denominators. It is then natural to look at what happens at the poles of $y^{-1}$, that is the points $P_i = (\alpha_i, 0)$, with $\alpha_i$ a root of $\defpollift$ in $\Z_q$ 
(note that $\alpha_i$ is a simple root of $\defpollift$), in order to deduce what happens globally.

Let $Q~=~\sum_{j = 1}^{k-1}{Q_j(x) \frac{\tau^j}{y^{\ell}}}$ be such that $R(x) \tau^k \frac{dx}{y^{\ell} } = S(x) \frac{dx}{y^{\ell}} + dQ$, with $Q_j$ defined over $\Q_{q}$ 
of degree less than $d$, for any $1 \leq j < k$.
As $\alpha_i$ is a simple root of $\defpollift$, the function $y$ is a uniformizing parameter for the local ring $\mathcal{O}_{P_i}$, that is, the ring of functions on $\curve$ regular at $P_i$.
Thus, the weak completion $\mathcal{O}_{P_i}^{\dagger}$ of $\mathcal{O}_{P_i}$ is $\Z_q \left[ \left[ y \right] \right]$. We can then write $dQ$ and $Q$ as series:

$$ dQ = \sum_{j \geq -\left(r(k-1)+ \ell +1\right)}{c_j y^{j} dy}$$
 and $$Q = \sum_{j \geq -\left(r(k-1)+ \ell +1\right)}{\frac{c_j}{j+1}y^{j+1}}.$$

As $c_j$ coincides with the corresponding coefficient of $R(x) \tau^k \frac{dx}{y}$ when $j < 0$, it lies in $\Z_q$. Hence, if we set 
$m = p^{\left\lfloor \log_{p}\left(r(k-1)+ \ell \right)\right\rfloor}$, then $m \frac{c_j}{j+1}$ is integral (ie in $\Z_q$) for $j < 0$.
Evaluating the coefficient of $y^{-\left(r(k-1)+ \ell \right)}$ at a pole $P_i$ of $y^{-1}$ in the expression $Q~=~\sum_{j = 1}^{k-1}{Q_j(x) \frac{\tau^j}{y^{\ell}}}$ 
gives $$Q_{k-1}(\alpha_i) = \frac{c_{-\left(r(k-1)+ \ell +1\right)}}{-\left(r(k-1)+ \ell \right)},$$
so $m Q_{k-1}(\alpha_i)$ is integral.
As this statement is independent of the point $P_i$ chosen, it holds for any $1 \leq i \leq d$. 
Thus, $m Q_{k-1}$ (of degree less than $d$) is integral at each of the $d$ distinct poles of $y^{-1}$ and it follows that 
$m Q_{k-1}$ is a polynomial defined over $\Z_{q}$.
The same argument applied to $Q - Q_{k-1} \tau^{k-1}$ gives that $m Q_{k-2}$ is a polynomial defined over $\Z_{q}$. 
By induction, all the $mQ_k$ are defined over $\Z_{q}$.
It follows that $mQ$ is integral and then $mS$ is as well because $S(x) \frac{dx}{y} = R(x) \frac{dx}{y} - dQ$.
\end{proof}

\begin{prop}\label{lemma:bound2}
If $S$ is a polynomial defined over $\Z_{q}$ with degree $m \geq d-1$, then the second step in Algorithm \ref{algo:reductionSuperelliptic} transforms 
$S(x) \frac{dx}{y^{\ell}}$ into $T(x) \frac{dx}{y^{\ell}}$, where $T$ is a polynomial defined over $\Q_{q}$ with degree less than $d-1$.
The coefficients of $T$ have denominator bounded by $p^{\left\lfloor \log_{p} \left( \frac{r(m + 1) - \ell d}{\delta}\right) \right\rfloor}$.
\end{prop}

\begin{proof}
We follow the approach of \cite[Lemma 4.3.5]{edixhoven}.

During the second step of Algorithm \ref{algo:reductionSuperelliptic}, we apply \eqref{eqn:Red2} several times.
As we divide by the leading coefficient of the polynomial given in \eqref{eqn:Red2}, positive powers of $p$ may occur at the denominators, which corresponds to positive powers of $x$ 
appearing in this step.
Hence we study what happens at the poles of $x$, that is the points at infinity $P_{\infty, k} = [1: \zeta_r^k: 0]$ of $\curve$, with $1 \leq k \leq \delta$.

Let $v_{\infty, k}$ denote the valuation at $P_{\infty, k}$.
Let $Q = \sum_{i = d-1}^m {ra_ix^{i-d+1}y^{r-\ell}}$ be such that $S(x) \frac{dx}{y^{\ell}} = T(x) \frac{dx}{y^{\ell}} + dQ$.
Then $v_{\infty, k}(x) = -\frac{r}{\delta}$, $v_{\infty, k}(y) = -\frac{d}{\delta}$ so ${v_{\infty, k}(dx) = -\frac{r+\delta}{\delta}}$. Moreover, 
$v_{\infty, k}\left(\frac{dx}{y^{\ell}}\right)~=\frac{\ell d - r - \delta}{\delta}$ so 
$$ v_{\infty, k}(Q) \geq  \frac{\ell d - r(m + 1)}{\delta}$$
and $$v_{\infty, k}\left(T\frac{dx}{y^{\ell}}\right)~\geq~\frac{\ell d - r(d-1)-\delta}{\delta}.$$
Let $z_k$ be a local uniformizer at $P_{\infty, k}$, so that $\mathcal{O}_{P_{\infty, k}}$ is $\Z_{q}[[z_k]]$.
Then we have the following expansion in $\Z_{q}[[z_k]]$: 
$$dQ = \sum_{j \geq \frac{\ell d - r(m + 1)}{\delta} - 1}{c_j z_k^jdz_k}$$ and $$Q = \sum_{j \geq \frac{\ell d - r(m + 1)}{\delta} -1}{\frac{c_j}{j+1}z_k^{j+1}}.$$

Let $m = p^{\left\lfloor \log_{p} \left( \frac{r(m + 1) - \ell d}{\delta}\right) \right\rfloor}$. 
 As $v_{\infty, k}(T\frac{dx}{y^{\ell}}) \geq - \frac{r(d-1) - \ell d + \delta}{\delta}$, then the $c_j$ are in $\Z_q$ for 
${j \leq -(r(d-1) - \ell d + \delta) / \delta - 1 = - (d(r - \ell ) + 2\delta - r) / \delta}$ since they coincide with the coefficients of $S$, so 
$m \frac{c_j}{j+1}$ is in $\Z_{q}$ for $j \leq - \frac{(d(r - \ell ) + 2\delta - r)}{\delta}$.

As all the $v_{\infty, k}(x^i \frac{dx}{y^{\ell}})$ are distinct and less than 
${\frac{\ell d - r-\delta -r(d-1)}{\delta} \leq - \frac{d(r - \ell ) + 2\delta - r}{\delta}}$ for $d-1 \leq i \leq m$ and $\ell>0$, it follows that
all the terms $m Q_j$ are in $\Z_{q}[x]$.
 
Since the previous statements are independent of $k$, they hold for any point at infinity of $\curve$. Hence 
$mQ$ is in $\Z_{q}[x]$, and $mT$ is in $\Z_{q}[x]$ too. 
\end{proof}

The two previous propositions estimate the loss of precision resulting from the reductions.
Recall that we are working with $p$-adic elements up to precision $\padicprec$, and that every element of $\funcringdagger$
is an overconvergent series, that is a series whose coefficients have $p$-adic valuation which grows at least linearly.
This means that for any $0 \leq i  \leq d - 2$ and $1 \leq j \leq r-1$, $\frobliftg \left( \frac{x^i dx}{y^j} \right)$ is a power series of the form $\sum_{k \geq 0}{F_k \tau^k}$ and
there exists $\yadicprec$ such that $v_p(F_k)$ is greater than $p^{\padicprec}$, for all $k > \yadicprec$.
Thus in practice, all the computed series are in fact polynomials of degree at most $\yadicprec$ in $\tau$.

\begin{proof}[ Proof of Theorem \ref{theo:bounds}]

The Weil bounds state that, if we put $$\padicprecini =  \left\lceil \log_{p} \left( 2\binom{2g}{g} q^{g/2} \right) \right\rceil, $$ then $\padicprecini$ is the minimal precision we 
have to take to compute the Weil polynomial exactly.
We also have to take into account all the digits lost by the divisions done during the reduction steps.

Let $\padicprec$ be the total precision we must take to compute the zeta function exactly, and $\padicprecint$ the intermediate precision we must take to do the first reduction step up to precision $\padicprecint$.
We will determine $\padicprec$ as a function of $\padicprecint$ and $\padicprecint$ as a function of $\padicprecini$ to recover $\padicprec$ as a function of $\padicprecini$.

Let $\yadicprec$ denote the integer such that the $v_p(F_k)$ are greater than $\padicprec$ for $k > \yadicprec$ (they are zero modulo $p^{\padicprec}$).

Using the expression of the action of Frobenius on vectors of $B$ given by \eqref{eqn:frobdiff}, we find that the integer $\yadicprec$ we want to determine is such that $\frac{k-a}{p} + 1 \geq \padicprec$ for $ k > \yadicprec$ and $\frac{\yadicprec - a}{p} + 1 < \padicprec$, so
$$\yadicprec = p(\padicprec - 1) + a - 1.$$

Note that $a<p$, since $a$ is the quotient in the division of $jp$ by $r$, and that $j \leq r-1$. 
Hence, $\yadicprec < p\padicprec - 1$.

\vspace{1ex}
To determine $\padicprec$, let us have a look at what happens during the reductions.
Proposition \ref{lemma:bound1} says that we lose $\left\lfloor \log_p\left(r(k-1)+\ell \right)\right\rfloor$ digits of precision during the first step of Algorithm \ref{algo:reductionSuperelliptic}, 
when we reduce $Q_k \tau^k \frac{dx}{y^{\ell}}$ with $\deg(Q_k) < d$. 
We want to take $\padicprec$ as small as possible such that after this first step of reduction, there remains $\padicprecint$ 
digits of precision for the second step of Algorithm \ref{algo:reductionSuperelliptic}, 
that is, such that
\begin{equation}\label{inequality}
\textstyle
\frac{k-a}{p} + 1 - \left\lfloor \log_p\left(r(k-1)+ \ell \right)\right\rfloor \geq \padicprecint \quad \mbox{ for } k > \yadicprec
\end{equation}
(here $\padicprec$ appears in the expression of $\yadicprec$).

The function $g: [\yadicprec + 1 \;, \; + \infty) \rightarrow \R$ mapping $k$ to the left hand side of Inequality \eqref{inequality} is strictly increasing; so we take the smallest $\padicprec$ such that $g(\yadicprec + 1) \geq \padicprecint$.
We find that $\padicprec$ is the minimal integer satisfying
\begin{equation} \label{firstpart}
\padicprec - \left\lfloor \log_p\left(p\left(r\padicprec-1\right) - r\right) \right\rfloor \geq \padicprecint.
\end{equation}

\vspace{1ex}
In order to determine $\padicprecint$, consider what happens during the second step of Algorithm \ref{algo:reductionSuperelliptic}: 
the terms contributing to the loss of precision during this step are those with degree $0$ in $\tau$, that is, the polynomials with degree in $x$ greater than $d-2$. 
For $1 \leq i \leq d-2$ and $1 \leq j \leq r$, we have
$$\frobliftg \left( x^i \frac{dx}{y^j} \right) = \sum_{k \geq 0}{\binom{-j/r}{k}p^{k+1}E^kx^{p(i+1)-1} \tau^{pk + a} \frac{dx}{y^{\ell}}}$$
and $i = d-2$ at worst, so $\deg(E^kx^{p(d-1)-1}) \leq k(pd-1) + p(d-1)-1$.
As $\frac{k(pd-1) + p(d-1)-1}{d} < (k+1)p$, we can write
$$\binom{-j/r}{k}p^{k+1}E^kx^{p(i+1)-1} \tau^{pk + a} = \sum_{a - p < j < pk + a}{c_{j, k}\tau^j},$$
so the degree in $x$ of the constant term of $\frobliftg \left( x^i \frac{dx}{y} \right)$ is at most $d~(p~-~a)$. 

Proposition \ref{lemma:bound2} says that the number of digits lost during this second step is $${\left\lfloor \log_{p}\left(\frac{r\left(d(p-a)+1\right)-\ell d}{\delta}\right)\right\rfloor}.$$ 
Since $jp~=~ar~+~\ell$, we lose ${\left\lfloor \log_{p}\left(\frac{dp(r-j)+r}{\delta}\right)\right\rfloor} \leq {\left\lfloor \log_{p}\left(\frac{dp(r-1)+r}{\delta}\right)\right\rfloor}$ digits.
Hence, 
\begin{equation}\label{secondpart}
 \padicprecint = \padicprecini + {\left\lfloor \log_{p}\left(\frac{r(dp+1)-d}{\delta}\right)\right\rfloor}.
\end{equation}

Let us put these two parts together. Combining Equations \eqref{firstpart} and \eqref{secondpart}, we should take
\begin{align*}
 \padicprec & = \min_{ n \in \N} \left\{n - \bigg\lfloor \log_p \Big( p(rn-1) - r \Big)  \bigg\rfloor  \geq \padicprecint \right\} \\
 & = {\min_{ n \in \N}} \left\{n - \bigg\lfloor \log_p \Big( p(rn-1) - r \Big)  \bigg\rfloor  \geq \padicprecini + \bigg\lfloor \log_{p}\Big(\frac{dp(r-1)+r}{\delta}\Big)\bigg\rfloor \right\}.
\end{align*}
\end{proof}

\section{The choice of the set of differentials}\label{sec:basis}
In Step 5 of Algorithm \ref{algo:cover} we compute the norm $M$ of the matrix of Frobenius $M_{\frobliftg}$ with respect to $B$, and in Step 6 we compute its characteristic polynomial $\chi(t)$.
If $M_{\frobliftg}$ has coefficients with denominators (coefficients in $\Q_q \setminus \Z_q$), then it is difficult to control the valuation of these denominators in
 the norm, and worse, we have to enlarge the precision
bounds to recover the Weil polynomial exactly by a number of digits that is hard to estimate.

So, it would be ideal if $M_{\frobliftg}$ was guaranteed to have coefficients in $\Z_q$.
Proposition \ref{prop:integral} tells us whether $M_{\frobliftg}$ has integral coefficients or not.

\begin{prop}\label{prop:integral}
 Let $0 \leq i \leq d-2$ and $1 \leq j \leq r-1$.
\begin{itemize}
\item The first step in Algorithm \ref{algo:reductionSuperelliptic}, applied to $\frobliftg(x^i \frac{dx}{y^j})$, computes 
a differential 
form whose coefficients have denominators of valuation bounded by ${\left\lfloor \log_p(r)\right\rfloor}$. 
\item The second step in Algorithm \ref{algo:reductionSuperelliptic}, applied to $\frobliftg(x^i \frac{dx}{y^j})$, computes a differential form whose coefficients have 
denominators of valuation bounded by 
${\left\lfloor \log_p\left( \frac{2 g + (\delta - 2)}{\delta}\right)\right\rfloor}$.
\end{itemize}
 \end{prop}

\begin{proof}
In this proof, we follow the approach of \cite[Lemma 3.4]{harrison}.

Recall that 
$$
\begin{array}{rcl}
 \frobliftg \left( x^i \frac{dx}{y^j} \right) &=& x^{p(i+1) - 1}\sum_{k \geq 0}{\binom{-j/r}{k}E^kp^{k+1}\tau^{pk + a}} \frac{dx}{y^{\ell}} \\
 & = & x^{p(i+1) - 1}\sum_{k \geq a}{p^{\frac{k-a}{p} + 1}Q_k(x) \tau^{k}} \frac{dx}{y^{\ell}},\\
\end{array}
$$
where $jp = ar + \ell.$
Lemma \ref{lemma:bound1} shows that after the first step in Algorithm \ref{algo:reductionSuperelliptic} the valuation in $p$ of the coefficients have the form:
\begin{equation}\label{eq:valuationfirststep}
k+1 - {\left\lfloor \log_{p}\left(r(pk + a -1)+ \ell \right)\right\rfloor} \geq k - {\left\lfloor \log_{p}\left(kr + j \right)\right\rfloor}.
\end{equation}
Let $g$ be the function defined on $[0, + \infty)$ by $g(x) =  k - \log_{p}\left(kr + j \right)$. Since $g$ is strictly increasing, the right hand side is maximal 
when $k=0$, which implies that the left hand side in Inequality \ref{eq:valuationfirststep} is greater than $\lfloor \log_p(j) \rfloor \leq \lfloor \log_p(r) \rfloor$.

Now consider the terms $p^{\frac{k-a}{p} + 1}Q_k(x)x^{p(i+1) - 1} \tau^k \frac{dx}{y^{\ell}}$ appearing in 
the second step of Algorithm \ref{algo:reductionSuperelliptic}. After 
normalizing the coefficients, each term will be expressible in the form $S\frac{dx}{y^{\ell}}$, with $$\deg(S) = p(i+1)-1 + \deg(Q_k) - dk.$$

After the second step in Algorithm \ref{algo:reductionSuperelliptic}, the denominators are bounded by 
$$ A = \left\lfloor \log_p\left(\frac{rp(i+1)+r \deg(Q_k) - rkd -\ell d}{\delta}\right)\right\rfloor - 1 - \frac{k-a}{p} $$
Since $Q_a = 1$ and $\deg(Q_k) < d$ for any $k > a$, this expression is maximal when $k = a$, which gives
$$ A \leq \left\lfloor \log_p\left(\frac{rp(i+1) - (ra + \ell )d}{\delta}\right)\right\rfloor - 1.$$

As $jp = ar+\ell$, we can express the right hand side of the inequality as
$ \left\lfloor \log_p\left(\frac{r(i+1) - jd}{\delta}\right)\right\rfloor$, which is less than $\left\lfloor \log_p\left(\frac{d(r-1) - r}{\delta}\right)\right\rfloor.$
Relation \ref{eq:genus} allows us to replace $d(r-1)$ with $2g + (\delta - 2)$.

Hence, the denominators of the coefficients of $M_{\frobliftg}$ are bounded by $$\left\lfloor \log_p\left(\frac{2g-(\delta-2)}{\delta}\right)\right\rfloor.$$
\end{proof}
This proposition tells us that if $d > r$, then the denominators mostly come from the second step of Algorithm \ref{algo:reductionSuperelliptic}.
In this case, if $p \geq \frac{2g - (\delta - 2)}{\delta}$, Algorithm \ref{algo:cover} gives us a matrix $M_{\frobliftg}$ with integral coefficients, while when $p < \frac{2g - (\delta - 2)}{\delta}$ they have denominators bounded by 
$\left\lfloor \log_p\left( \frac{2g - (\delta - 2)}{\delta} \right)\right\rfloor$.
So if we want to minimize these denominators, we should find a set of differentials which avoids this second step in Algorithm \ref{algo:reductionSuperelliptic}.

\subsection{Another set of differentials which avoids the second step in Algorithm \ref{algo:reductionSuperelliptic}}
Consider the set of differentials
$$B' = \left\{ x^i \frac{dx}{y^{r+j}} \; | \; i \in [0, \ldots, d-2 ] , j \in [1, \ldots, r-1] \right\}.$$

The reduction formulae in \S \ref{sec:adaptation} tells us that $B'$ spans $\HMW{1}$. We will show in Proposition \ref{prop:newbasis} that $\ideal{B'}$ decomposes into a direct sum of two 
subspaces including $\HMW{1}$, so we can recover the Weil polynomial of $\curve$ from the action of Frobenius acting on $B'$. Moreover, the matrix of Frobenius with respect to $B'$ has the 
same 
block structure as the matrix of Frobenius with respect to $B$.

We will show in Proposition \ref{prop:newbasis}  that when $\delta = 1$, then $B'$ is a basis of $\HMW{1}$. Thus, we can do the computations with $B'$ instead of $B$ in the Gaudry--G\"urel algorithm.

The following theorem tells us that doing the computations with $B'$ avoids the second step in Algorithm \ref{algo:reductionSuperelliptic}.
\begin{theorem}\label{prop:otherbasis}
Let $0 \leq i \leq d-2$ and $1 \leq j \leq r-1$.
The first step in Algorithm \ref{algo:reductionSuperelliptic}, applied to $\frobliftg\left( x^i \frac{dx}{y^{j+r}} \right)$, gives a form which is a linear combination of elements of $B'$ and
 whose coefficients have denominator bounded by $p^{\left\lfloor \log_p(2r-1)\right\rfloor}$.
\end{theorem}
\begin{proof}
In this proof, we follow the approach of \cite[Lemma 3.4]{harrison}.

Let $0 \leq i \leq d-2$, $ 1 \leq j \leq r-1$ and $\varepsilon = 0$ or 1. Using the same calculation as we did to get \eqref{eqn:frobB}, we have
$$ \frobliftg \left( x^i \frac{dx}{y^{\varepsilon r+j}} \right)=  x^{p(i+1) - 1}\sum_{k \geq 0}{\binom{-(\varepsilon r +j)/r}{k}p^{k+1}E^k(x) \tau^{p(k + \varepsilon) + a - \varepsilon}} \frac{dx}{y^{\varepsilon r + \ell }},$$
where $E = \frac{\frobliftg(\defpollift(x)) - \defpollift(x)^{p}}{p}$ and $jp = ar + \ell$.
After normalizing (using the equation of $\curve$ to remove all the terms with degree in $x$ greater than $\deg(\defpol)$), and using the fact that $Q_{a + \varepsilon (p-1)} = 1$, we have
\begin{equation}\label{eqn:frobB'}
 \frobliftg \left( x^i \frac{dx}{y^{\varepsilon r +j}} \right) = \sum_{k \geq k_0}{p^{\frac{k-a + \varepsilon}{p} + 1 - \varepsilon} R_k(x) \tau^{k}} \frac{dx}{y^{\varepsilon r + \ell }},
\end{equation}
with $k_0 \geq a + \varepsilon (p-1) - \left\lfloor \frac{p(i+1) - 1}{d} \right\rfloor.$

When $\varepsilon = 1$, we have $k_0 \geq a + (p-1) - \left\lfloor \frac{p(i+1) - 1}{d} \right\rfloor$. The right hand side is minimal when $i = d-2$, with value $a + p-1 - (p-1)  = a \geq 0$.
This means that when $\varepsilon = 1$, the only term which may be reduced by the second step of Algorithm~\ref{algo:reductionSuperelliptic} is the first one, which has degree 0. 
So the first step in 
Algorithm~\ref{algo:reductionSuperelliptic}, applied to $\frobliftg \left( x^i \frac{dx}{y^{j+r}} \right)$ gives a form which is a linear combination of elements of $B'$.

Consider the terms ${\binom{-(r +j)/r}{k}p^{k+1}E^k \tau^{p(k + 1) + a - 1}}$ appearing in the first step of Algorithm \ref{algo:reductionSuperelliptic}.
Lemma \ref{lemma:bound1} shows that after the first step in Algorithm \ref{algo:reductionSuperelliptic} the valuation in $p$ of the coefficients is
$$k+1 - {\left\lfloor \log_{p}\left(r(p(k+1) + a -2)+ \ell \right)\right\rfloor} \geq k - {\left\lfloor \log_{p}\left(r(k+1)+j \right)\right\rfloor}$$
An easy calculation shows that the right hand side term is maximal when $k=0$ and has value $\lfloor \log_p(r+j) \rfloor \leq \lfloor \log_p(2r-1) \rfloor$.
\end{proof}

If $p \geq 2r$, then the computations done in Algorithm \ref{algo:cover} using $B'$ will give a matrix with integral coefficients. If however $p < 2r$, then the 
matrix computed with respect to $B'$ will have coefficients with denominators bounded by $p^{\left\lfloor \log_p(2r-1)\right\rfloor}$. 

At the end of this subsection, we will compare the two sets of differentials $B$ and $B'$ and provide a criterion to recommend which one to use.

 As there is no second step of reduction when we use $B'$ in the computations, we can slightly reduce the bounds on precision as follows.
\begin{theorem}
In order to compute the Weil polynomial of $\curve$ exactly using Algorithm \ref{algo:cover} with the basis $B'$, we can take
$$\padicprec = \min_{ n \in \N} \left\{n - \bigg\lfloor \log_p \Big( pr(n+1) - 3r \Big)  \bigg\rfloor  \geq \padicprecini \right\}$$
where
$$\padicprecini = \left\lceil \log_{p} \left( 2\binom{2g}{g} q^{g/2} \right) \right\rceil.$$
\end{theorem}
\begin{proof}
We follow the approach of the proof of Theorem \ref{theo:bounds}.

Using \eqref{eqn:frobB'}, we find that
$$\yadicprec = p \padicprec + a -2.$$

As the second step in Algorithm \ref{algo:reductionSuperelliptic} is not executed, we only need to take into account the loss of 
precision resulting from the first step in Algorithm \ref{algo:reductionSuperelliptic}. 
So we want to find $\padicprec$ as small as possible such that after the first step of reduction, there remains $\padicprecini$ digits of precision, that is
\begin{equation}
\textstyle
\frac{k-a + 1}{p} - \left\lfloor \log_p\left(r(k-1)+ \ell \right)\right\rfloor \geq \padicprecini \; \mbox{ for } k > \yadicprec
\end{equation}
(here $\padicprec$ appears in the expression of $\yadicprec$).
The function $g: [\yadicprec + 1 \;, \; + \infty) \rightarrow \R$ which maps $k$ to the left hand side of Inequality \eqref{inequality} is strictly increasing; so we take the smallest $\padicprec$ such that $g(\yadicprec + 1) \geq \padicprecint$.
We find that $\padicprec$ is the minimal integer satisfying
\begin{equation}
\padicprec - \left\lfloor \log_p\left(pr(\padicprec+1) - 3r\right) \right\rfloor \geq \padicprecini,
\end{equation}
which ends the proof.
\end{proof}

\vspace{1ex}
Finally we observe that if $p \geq 2r$, then using $B'$ instead of $B$ in the computations is much better for two reasons. First, it 
always lead to an integral matrix. Second, as the second step of Algorithm \ref{algo:reductionSuperelliptic} is not executed, we do fewer operations 
and at a lower precision (because we save near $\bigg\lfloor \log_{p}\Big(\frac{dp(r-1)+r}{\delta}\Big)\bigg\rfloor$ digits of precision) so it slightly decrease the complexity of Algorithm \ref{algo:cover}.

Otherwise, it is not always possible to have a matrix with integral coefficients and we have to check if 
$$\max{ \left(\lfloor \log_p(r) \rfloor, \Big\lfloor \log_p\left(\frac{2g - (\delta - 2)}{\delta}\right) \Big\rfloor\right)}$$
is greater than $$\lfloor \log_p(2r-1) \rfloor.$$
If it is the case, then we use $B'$ and if not, we have to choose between $B$ and $B'$, taking into account that we will 
have to enlarge the bounds on precision (which are smaller using $B'$) due to the computation of $M = M_{\frobliftg} \cdot M_{\frobliftg} ^{\sigma}  \cdots M_{\frobliftg}^{\frobp^{n-1}}$ and its characteristic polynomial.

\subsection{How does the use of $B'$ change Algorithm \ref{algo:cover} ?}
Using $B'$ instead of $B$ does not significantly change Steps 1 to 5 of our algorithm, where we compute the action of Frobenius on cohomology: 
the operations are the same, only some powers of $p$ and some indexes in the expressions change.
However, the extra factor $U$ in Step 6 of Algorithm \ref{algo:cover} may be slightly different.

\begin{prop}\label{prop:newbasis}
Let $\eta : \ideal{B'} \rightarrow \HMW{1}$ be the map sending differentials to their classes in cohomology.
\begin{itemize}
 \item If $\delta = 1$, then $\eta$ is an isomorphism.
 \item Otherwise, $\eta$ has a $\delta - 1$-dimensional kernel stable under the action of Frobenius, and there exists a basis $W$ of $\ker(\eta)$ such that the matrix of the $q$-th power 
Frobenius with respect to $W$ is a generalized permutation matrix: that is, in the form 
$DP$ where $D$ is a diagonal matrix and $P$ is a permutation matrix.
\end{itemize}
\end{prop}
\begin{proof}
We follow the approach of \cite[Proposition 3.5]{harrison}.

The reduction formulae in \S \ref{sec:adaptation} tells us that $B'$ spans $\HMW{1}$ so 
$\eta$ is surjective and $\eta$ has a kernel of dimension $$ \dim(\mbox{span}(B')) - \dim( \HMW{1}) = \delta - 1.$$
It follows that if $\delta=1$, then $\eta$ is an isomorphism.

\vspace{1ex}
Otherwise, let $$\omega = \sum_{j = 1}^{r-1} \left( \sum_{i = 0}^{d-2} \lambda_{i, j} x^i  \right) \frac{dx}{y^{j+r}} = \sum_{j = 1}^{r-1} {V_j(x) \frac{dx}{y^{j+r}}}$$
be a non zero element of $\ker(\eta)$, that is, 
such that $\omega \equiv 0$ in $\HMW{1}$. We want to determine $\ker(\eta)$, that is, to describe $\omega$ explicitly.

Using the first reduction formula, we get 
\begin{align*}
 0 \equiv \omega \equiv \sum_{j = 1}^{r-1}{\left( A_j + \frac{r}{j} B_j' \right) \frac{dx}{y^j}}, \mbox{ where } V_j = A_j f + B_j f'.
\end{align*}
As the degree of $V_j$ is lower than $d-1$ for $1 \leq j \leq r-1$, the degrees of $A_j$ and $B_j$ are lower to $d-1$ too, 
and there is no second reduction, so we have expressed $\omega$ as a linear combination of elements of $B$.
Since the elements of $B$ are linearly independent, it follows that 
$$ A_j + \frac{r}{j} B_j'  = 0  \; \mbox{for all } 1 \leq j \leq r-1,$$
so $\omega$ has the form
\begin{align*}
\omega = \sum_{j = 1}^{r-1}{\left( - \frac{r}{j} B_j' f + B_j f' \right) \frac{dx}{y^{r+j}}}.
\end{align*}

In order to describe $B_j$ for $1 \leq j \leq r-1$, let $j$ be an integer such that $1 \leq j \leq r-1$.
If we write $B_j = a_{b_j} x^{b_j} + \cdots$ with $b_j \geq 0$ and $a_{b_j} \neq 0$, then the leading term in $V_j$ is $\left( d - \frac{r}{j} b_j \right)a_bx^{d-1+b_j}$. Since $\deg(V_j) \leq d-2$, we have that 
 $b_j = \frac{jd}{r}$. Normalizing $B_j$ so that its leading coefficient is 1, it follows easily that its lower coefficients are completely determined inductively by the condition on
 $\deg(V_j)$. 
Hence 
$$ \omega = \sum_{1 \leq j \leq r-1, \frac{r}{\delta} | j} \lambda_j \omega_j, $$
where $$\omega_j = V_j \frac{dx}{y^{r+j}}$$ and $j$ is a multiple of $\frac{r}{\delta}$: that is, where $j = \alpha \frac{r}{\delta}$ with 
$1 \leq \alpha \leq \delta -1$. We may take 
$$\omega_j = \left( B_j f'- \frac{r}{j} B_j' f \right) \frac{dx}{y^{r+j}} = d \left( -\frac{r}{j} \frac{B_j}{y^j} \right),$$
 with $\deg(B_j) = \frac{jd}{\delta}$ .

Note that $\ker(\eta)$ is generated (as a vector basis) by the $\omega_j$ for $1 \leq j \leq r-1$. 
Indeed, the few lines above show that the $\omega_j$ lie in the kernel 
of $\eta$ and the first reduction formula shows that they are linearly independent. Since there are $\delta - 1$ of them, 
which is the dimension of $\ker(\eta)$, it follows that they form a basis of $\ker(\eta)$.

\vspace{1ex}
Now we determine the action of Frobenius on $\ker(\eta)$.
Recall that $\frobliftg$ denotes the $p$-th power Frobenius, so $\frobkliftg{n}$ is the $q$-th power Frobenius.

Let $\omega_j = d \left( -\frac{r}{j} \frac{B_j}{y^j} \right)$ be one of the generators of $\ker(\eta)$. 
The Frobenius commutes with the operator $d$, so $\ker(\eta)$ is stable
under the action of the Frobenius.
This means that $$\frobkliftg{n} (\omega_j) = \lambda_j \omega_l - r dQ, \; \mbox{ where } \ell = jp \bmod r \mbox{ and } Q = \sum_{k \geq 2}{Q_k \frac{\tau^{k-1}}{y^{\ell}}},$$
so $$\frac{1}{j} d \left( \frobkliftg{n}\left(\frac{B_j}{y^j}\right)\right) = \frac{\lambda_j}{\ell} d \left( \frac{B_l}{y^{\ell}}\right) + dQ.$$
This implies
$$\frac{1}{j} \frobkliftg{n}\left(\frac{B_j}{y^j}\right) = \frac{\lambda_j}{\ell} \frac{B_l}{y^{\ell}} + \left( \frac{Q_2}{y^{r+ \ell }} + \frac{Q_3}{y^{2r+ \ell }} + \cdots \right).$$

We compute $\frobliftg \left(\frac{B_j}{y^j}\right) = B_j(x^p) \frobliftg(y)^{-j}$: as $B_j(x^p) = x^{b_jp} + \cdots$, we have (after normalization)
$$B_j(x^p) = u_{a}(x)y^{ar} + \cdots + u_0(x),$$
where $a  = \frac{jp - \ell}{r} + \lfloor\frac{\ell}{d} \rfloor$ and $u_a$ is monic of degree $\frac{d}{\delta} \times (\ell  \bmod d)$ and $\deg (u_i)<d$, for $0 \leq i \leq a$.
Indeed, $\deg(B_j(x^p)) = \frac{jpd}{\delta} = \frac{(ar+ \ell )d}{\delta} =  d \left( ar + \frac{\ell}{ \delta} \right)$ so this gives the normalization above.
Thus,
\begin{align*}
\frobliftg \left(\frac{B_j}{y^j}\right) & =  \left( u_{a}(x)y^{ar} + \cdots + u_0(x) \right)\left( y^{-jp}(1 + \frac{v_1}{y^r} + \frac{v_2}{y^{2r}} + \cdots) \right) \\
& = \frac{u_a(x)}{y^{\ell}} + \left( \frac{a_1}{y^{r+ \ell }} +  \frac{a_2}{y^{2r+ \ell }} + \cdots\right),
\end{align*}
with $u_a$ of degree $\frac{\ell d}{\delta}$ and $\ell = jp \bmod r$.

Proposition \ref{prop:otherbasis} shows that Frobenius applied to $\left( \frac{a_1}{y^{r+ \ell }} +  \frac{a_2}{y^{2r+ \ell }} + \cdots\right)$ has the same form: 
that is, it has no term in $y^{-\ell }$.
 This means that 
$$ \frobkliftg{n} \left( \frac{B_j}{y^j}\right)  =  \frac{B_m(x)}{y^m} + \left( \frac{Q_1}{y^{r+m}} +  \frac{Q_2}{y^{2r+m}} + \cdots\right),$$
where $m = jq \bmod r$ and $B_m$ is monic of degree $\frac{md}{\delta}$, so $\lambda_j = \frac{m}{j}$ with $m = jq \bmod r$.

Let $\phi$ be the permutation of $\{ \frac{r}{\delta},2 \cdot \frac{r}{\delta},  \cdots, (\delta-1) \cdot \frac{r}{\delta} \}$ defined by $$\phi: j \longmapsto jq \bmod r$$
(this is well defined because it acts on a set isomorphic to the subgroup of elements of order dividing $\delta$ in $\F_r^*$ and this is a permutation because $q$ is prime to $r$). 
We have $$\frobkliftg{n}(\omega_j) = \frac{\phi(j)}{j} \omega_{\phi(j)},$$ so the matrix of Frobenius with respect to our basis $W = (\omega_1, \cdots, \omega_{\delta-1})$ 
of the kernel of 
$\eta$ has the form $DP$ with D diagonal and $P$ a permutation matrix.
\end{proof}

The previous proposition allows us to describe the extra factor appearing in the characteristic polynomial of the $q$-th power Frobenius acting on $B'$.
\begin{theorem}
 The Weil polynomial $P$ of $\curve$ is
$$P(t) = \frac{\chi_{M}(t)}{U(t)}, $$
where $\chi_{M}(t)$ is the characteristic polynomial of the matrix corresponding to the $q$-th power Frobenius with respect to $B'$ and 
$$U(t) = \displaystyle{\prod_{i | \delta \;, \; i > 1 }(t^{k_i} - 1)^{\frac{\varphi(i)}{k_i}}},$$
where $k_i$ is the order of $q$ in $\Z / \varphi(i) \Z$ and $\varphi$ is the Euler totient function.
\end{theorem}
\begin{proof}
Proposition \ref{prop:newbasis} tells us that $B'$ decomposes into a direct sum of two spaces including $\HMW{1}$. This means that if we compute 
the characteristic polynomial of the Frobenius acting on $B'$, then we have to remove an extra factor from it to recover the Weil polynomial.
This extra factor $U(t)$ is the characteristic polynomial of the Frobenius acting on $\ker(\eta)$, whose action is described 
by Proposition \ref{prop:newbasis}. 
Since we can change the order of the vectors in the basis $W$ of $\ker(\eta)$ (see the Proof of Proposition \ref{prop:newbasis}) without changing the characteristic polynomial, 
$U$ is the characteristic polynomial of $\widetilde{D}\widetilde{P}$ where $\widetilde{P}$ is a block diagonal matrix of permutation 
(the blocks correspond to the disjoint cycles of the permutation $\phi$ acting on $\ker(\eta)$)
, so $\chi_{\widetilde{D}\widetilde{P}}$ 
is a product of polynomials of the form $$(t^k - \Delta_k), $$
where $\Delta_k$ is the determinant of the corresponding block of $\widetilde{D}\widetilde{P}$ and has the form
$$ \Delta_k = \frac{\phi(j)}{j} \times \frac{\phi^2(j)}{\phi(j)} \times \cdots \times \frac{\phi^{k}(j)}{\phi^{k-1}(j)}.$$ 
As $\phi^k(j) = j$, the previous product is telescopic so $\Delta_k=1$.

In order to describe the number and the lengths of these cycles, we show that the action of $\phi$ on $\ker(\eta)$ is equivalent to the action of the $q$-th power Frobenius on the 
$\delta-1$ points at infinity
$$P_{k, \infty} = [1 : \zeta_r^k : 0], \; \mbox{ for } 1 \leq k \leq \delta - 1.$$

The map 
$$\psi:\Big\{ P_{k, \infty} = [1:\zeta_r^k:0] \in \curve(\overline{\F_q}) : k \in [1, \delta-1]\Big\} \longrightarrow \left\{ j = k \frac{r}{\delta} \in \F_r^{\times} \; : \; k \in [1, \delta-1] \right\}$$ defined by 

$$\psi: P_{k, \infty}  \longmapsto  j = k \frac{r}{\delta}$$
is a bijection compatible with the action of the $q$-th power Frobenius and the action of $\phi$.
Indeed,
\begin{align*}
 \psi \left( \frobklift{n} (P_{k, \infty}) \right) \; & = \psi([1:\zeta_r^{kq \bmod r}:0]) \\
& = \psi([1:\zeta_r^{kq \bmod \delta}:0]) \;\mbox{ because of the relations in } \mathbb{P} \left( \frac{r}{\delta}, \frac{d}{\delta}, 1 \right) \\
& = (kq \bmod \delta) \times \frac{r}{\delta} \\
& = (kq \times \frac{r}{\delta}) \bmod r \; \mbox{ because if } kq = a \delta + b \mbox{ then } kq \frac{r}{\delta} = ar + b \frac{r}{\delta}  \\
&= \phi(k).
\end{align*}
This shows that $U$ is the characteristic polynomial of the $q$-th power Frobenius acting on the $\delta -1$ points at infinity of $\curve$ described above.

In $\mathbb{P}(\frac{r}{\delta}, \frac{d}{\delta} , 1)$, the points $[X:Y:Z]$ at infinity are those with $Z=0$ and $Y^r = X^d$, which is equivalent 
to $\widetilde{Y}^{\delta} = \widetilde{X}^{\delta}$, with $\widetilde{Y} = Y^{\frac{r}{\delta}}$ and $\widetilde{X} = X^{\frac{d}{\delta}}$: that is, 
$T^{\delta} = 1$ where $T = \frac{\widetilde{Y}}{\widetilde{X}}.$

As we have $$T^{\delta} - 1 \; = \; \prod_{i | \delta} \Phi_i(T) \; = \; (T - 1)\prod_{{i | \delta} \; , \; {i > 1}} \Phi_i(T),$$ 
where $\Phi_i$ is the $i$-th cyclotomic polynomial, the Frobenius acts on the $\delta - 1$ points at infinity  $P_{\infty, k}$ for $1 \leq k < \delta$ as it acts on the roots of
$$\prod_{i | \delta , i > 1} \Phi_i(T).$$

If $i$ is an integer dividing $\delta$, then the Frobenius acts on the $\varphi(i)$ roots of $\Phi_i(T)$ as a product of 
$\frac{\varphi(i)}{k_i}$ cycles of length $k_i$, where $k_i$ is the order of $q$ in $\Z / \varphi(i) \Z$. 
So the Frobenius map acts on these $\delta -1 $ points at infinity of $\curve$ as a permutation whose corresponding
permutation matrix has characteristic polynomial
$$\prod_{i | \delta, i > 1}(t^{k_i} - 1)^{\frac{\varphi(i)}{k_i}}.$$
\end{proof}

\section{Numerical experiments}\label{sec:an}
In this section, we give some experimental results obtained with a prototypical (and unoptimized) implementation of our algorithm 
in Magma 2.18.  The experiments were run on a single core of a Xeon E5520 machine (2.26GHz, 72GB RAM).
We tested our results by taking a random divisor on the curve, multiplying it by the supposed order of the Jacobian of the curve (which is the Weil polynomial evaluated at 1), and checking whether 
the resulting divisor is principal.

\begin{example}
We consider the genus 13 curve 
\begin{align*}
\curve_{3, 15}: y^3 \; & = x^{15} + (2 \alpha  + 5)x^{13} + 2 \alpha x^{12} +  \alpha x^{11} + (3 \alpha  + 6)x^{10} \\
& + 3x^9 + (2 \alpha  + 4)x^8 + 4 \alpha x^7 + 6 \alpha x^6 + 6x^4 \\
&+  \alpha x^3 + (4 \alpha  + 5)x^2 + (6 \alpha  + 5)x
\end{align*}
defined over $\F_{49} = \F_7[\alpha] / \ideal{\alpha^2 - \alpha + 4}$.
After 585 seconds, our implementation returns the Weil polynomial of $\curve$, whose Weil coefficients are
$$ a_1 = 4, \quad a_2 = -88, \quad a_3 =-317, \quad a_4 =3477, \quad a_5 = 45743, \quad a_6 = -38408, $$
$$ a_7 = - 3064081, \quad a_8 = 1826186, \quad a_9 = 105964107, \quad a_{10} = 178170657, $$
$$ a_{11} = - 3878128722, \quad a_{12} = - 10860792624 \quad \mbox{ and } a_{13} = 227741125446.$$
\end{example}

\vspace{1ex}
\begin{example}
We consider the genus 26 curve 
\begin{align*}
\curve_{5,15}: y^5 \; &=  x^{15} + (4\alpha + 7) x^{13} + (4 \alpha + 6) x^{12} + (2 \alpha + 4) x^{11} \\&
          + (10 \alpha + 4) x^{10} + (\alpha + 10) x^9 + 4 x^8 + 2 x^7 + 6 x^6 \\&
          + (3 \alpha + 1) x^5 + 10 x^4 + (10 \alpha + 1) x^3 + (5 \alpha + 9) x^2 \\&
          + (7 \alpha + 4) x + 2 \alpha + 6
\end{align*}
defined over $\F_{121} = \F_{11}[\alpha] / \ideal{\alpha^2 - \alpha + 4}$.
After 23922 seconds, our implementation returns the Weil polynomial of $\curve$ whose Weil coefficients are
$$ a_1 = 36, \; a_2 = 418, \; a_3 = 3928, \; a_4 = 107603, \; a_5 = 1546802, \; a_6 = 10195080, $$
$$ a_7 = 189193348, \; a_8 = 3908194517, \; a_9 = 35529836037, \; a_{10} = 323855056565, $$
$$ a_{11} = 6026279205222, \; a_{12} = 71054667707163, \; a_{13} = 577639402235514, $$
$$ a_{14} = 7788857330417489, \; a_{15} = 103362684561282136,$$
$$ a_{16} = 988282517113615745, \; a_{17} = 11354454883387292669, $$
$$_{18} = 122508522344304060111, \; a_{19} =999211815604433952646,$$
$$ a_{20} = 13694995222065645049886, \; a_{21} = 174130364097714846506217 $$
$$ a_{22} = 1066845743104788110404502, \;_{23} = 11897270459284483568657805, $$
$$a_{24} = 243759226939902383459526275, \; a_{25} = 1925128879480201238759308035, $$
$$\mbox{ and } a_{26} = 8130284653021215396447907725.$$
\end{example}

\begin{example}
We consider the genus 45 curve
    $$\curve_{11,11} : y^{11} = x^{11} + 21x^9 + 22x^8 + 12x^7 + 14x^6 + 5x^4 + 15x^3 + 6x^2 + 15x + 11$$
defined over $\F_{23}$.
After 159200 seconds, our implementation returns the Weil polynomial of $\curve$ whose Weil coefficients are:
$$ a_1 = -10, \; a_2 = 148, \; a_3 = -1172, \; a_4 = 11400, \; a_5 = -75082, \; a_6 = 583607, $$
$$ a_7 = -3423792, \; a_8 = 23458758, \; a_9 = -127681770, \; a_{10} = 815749654, $$
$$ a_{11} = -4274768142, \; a_{12} = 26177112830, \; a_{13} = -133290333147, $$
$$ a_{14} = 792181088309, \; a_{15} = -3931625501060, \; a_{16} = 22819266210165, $$
$$a_{17} = -110481821962459, \; a_{18} = 633740960651940, \; a_{19} = -3001343844798677, $$
$$ a_{20} = 17054767132345719, \; a_{21} = -79052006236498542,$$
$$ a_{22} = 445634829426753123, \; a_{23} = -2018975937263556165, $$
$$a_{24} = 243759226939902383459526275, \;  a_{25} = -50378603603766216893, $$
$$ a_{26} = 281146158641010525301, \; a_{27} = -1239849286459249269112, $$
$$a_{28} = 6921368868854435563991 , \; a_{29} = -30287237899941389111850, $$
$$ a_{30} = 168719424687252264076767 , \; a_{31} = -728584303024763825003860, $$
$$ a_{32} = 4051750456875540838493246 , \; a_{33} = -17207665565047921783353414, $$
$$ a_{34} = 95531537944645720980803334, \; a_{35} = -398515032624667404187154280, $$
$$ a_{36} = 2220486855862732905431832556, \; a_{37} = -9115467662197167357206988372, $$
$$ a_{38} = 50987572400077029250253058483, $$
$$ a_{39} = -207263506930883933858403922280, $$
$$ a_{40} = 1165874930218286023405099204275 ,$$
$$ a_{41} = -4712376446054941126784485443520, $$
$$a_{42} = 26631761506101496258816899274283, $$
$$ a_{43} = -107766534346210234686112282045945, $$
$$ a_{44} = 610647567000069960495606605432680, $$
$$\mbox{ and } a_{45} = -2472407143793335018389394336486111.$$
\end{example}

\begin{example}
We consider the genus 57 curve 
\begin{align*}
\curve_{7,21}: y^7 \; &=  x^{21} + \alpha^{166}x^{19} + \alpha^{12}x^{18} + \alpha^{64}x^{17} + \alpha^{102}x^{16} + \alpha^{166}x^{15} + 12x^{14} \\&
          + \alpha^{25}x^{13} + \alpha^{68}x^{11} + \alpha^{117}x^{10} + \alpha^8x^9 + \alpha^{15}x^8 + \alpha^{16}x^7 + \alpha^{127}x^6\\&
          +  \alpha^{90}x^5 + \alpha^{43}x^4 + \alpha^{128}x^3 + \alpha^{40}x^2 + \alpha^{125}x + \alpha^{99}
\end{align*}
defined over $\F_{169} = \F_{13}[\alpha] / \ideal{\alpha^2 - \alpha + 2}$.
After 380881 seconds, our implementation returns the Weil polynomial of $\curve$, which factors as
$$ P_{\curve_{7, 21}}(T) = (T+13)^6 P(T)^2, $$
where $P$ is the Weil polynomial of a 27-dimensional abelian variety whose Weil coefficients are:
$$ a_1 = 14, \; a_2 = 224, \; a_3 = 3804, \; a_4 = 68075, \; a_5 = 650370, \; a_6 = 8859458, $$
$$ a_7 = 72307214, \; a_8 = 1083567163, \; a_9 = -157139189, \; a_{10} = -20620569697, $$
$$ a_{11} = -2357957261121, \; a_{12} = -16670241272334, \;  a_{13} = -448263291116144, $$
$$ \; a_{14} = -145927900246555, \; a_{15} = -23388115214168173, \; $$
$$ a_{16} = 647629219619169060, \; a_{17} = 6886478186860664665, $$
$$ a_{18} = 328920785102728658021, \; a_{19} = 2370798532171844617115, $$
$$ a_{20} = 52751582248734601974196, \; a_{21} = 326749252127936392530802, $$
$$ a_{22} = 5641762316975885681964474,\; a_{23} = -35674382266353914319048358, $$
$$ a_{24} = -321587628360190547802537740, \; a_{25} = -19029290083673947265278225863, $$
$$ a_{26} = -152084904894251081055443498722, $$
$$ \mbox{ and } a_{27} = -3983383747839588680645320044353.$$
\end{example}

\bibliographystyle{amsalpha}
\def\cprime{$'$}
\providecommand{\bysame}{\leavevmode\hbox to3em{\hrulefill}\thinspace}
\providecommand{\MR}{\relax\ifhmode\unskip\space\fi MR }
\providecommand{\MRhref}[2]{%
  \href{http://www.ams.org/mathscinet-getitem?mr=#1}{#2}
}
\providecommand{\href}[2]{#2}

\end{document}